\documentclass[a4paper, 11pt]{article}

\usepackage{mathtools, amssymb, amsthm, bbm}
\usepackage{graphicx}
\usepackage{booktabs, multirow}
\usepackage{url}
\usepackage[colorlinks,citecolor=blue]{hyperref}
\usepackage[round]{natbib}
\usepackage{authblk}
\usepackage[margin = 1.25in]{geometry}
\usepackage{makecell}

\theoremstyle{plain} 
\newtheorem{theorem}{Theorem}[section] 
 
\newtheorem{proposition}[theorem]{Proposition}

\theoremstyle{definition}

\theoremstyle{remark}

\newcommand{\one}{\mathbbm{1}}
\newcommand{\myS}{{\rm S}}

\begin{document}

\title{Valid sequential inference on probability forecast performance}
\author{Alexander Henzi} 
\author{Johanna F.~Ziegel}
\affil{University of Bern, Switzerland \\ \hfill \\ \texttt{alexander.henzi@stat.unibe.ch} \ \texttt{johanna.ziegel@stat.unibe.ch}}
\date{\today} 
\maketitle

\begin{abstract}
Probability forecasts for binary events play a central role in many applications.
Their quality is commonly assessed with proper scoring rules, which
assign forecasts a numerical score such that a correct forecast achieves a 
minimal expected score. In this paper, we construct e-values for testing
the statistical significance of score differences of competing forecasts in sequential settings.
E-values have been proposed as an alternative to p-values for hypothesis
testing, and they can easily be transformed into conservative p-values by taking the multiplicative inverse. The e-values proposed in this article are valid in 
finite samples without any assumptions on the data generating processes. They also allow optional stopping, so a forecast user may decide to interrupt 
evaluation taking into account the available data at any time and still draw statistically valid inference, which is generally not true for classical p-value based tests. In a case study on postprocessing of precipitation forecasts, state-of-the-art forecasts dominance tests and e-values lead to the same conclusions.
\end{abstract}

\section{Introduction}
Consider a forecast user who compares probability predictions $p_t, q_t \in [0,1]$, $t \in \mathbb{N}$, for a binary event $Y_{t + h} \in \{0, 1\}$, where $h \geq 1$ is the time lag between the forecasts and observations. At time $t$, the forecasts $p_t, q_t$ as well as any predictions and observations before $t$ are known. This setting encompasses many practical situations such as probability of precipitation forecasts $h$ days ahead or predictions of negative economic growth in the next quarter. The forecast user wants to draw conclusions on the relative performance of $p_t$ and $q_t$, that is, identify the better of the two forecasts.

Probability forecasts for binary events are arguably the simplest and best understood type of probabilistic forecasts; see \cite{Winkler1996} for an earlier overview and more recent reviews in \cite{Gneiting2007a}, \cite{Ranjan2010} and \cite{Lai2011}. The key requirements for probability forecasts are calibration, meaning that events with a predicted probability of $p$ should occur at a frequency of $p$, and sharpness, which requires the forecast probabilities to be as informative as possible, i.e.~close to $0$ or $1$. These properties are simultaneously assessed with proper scoring rules \citep{Gneiting2007a}, which coincide with consistent scoring functions for the mean \citep{Gneiting2011} in the case of probability forecasts and will be simply referred to as scoring functions in this article. A scoring function $\myS = \myS(p, y)$ maps a forecast probability $p$ and an observation $y$ to a numerical score, with smaller scores indicating a better forecast. More precisely, $\myS$ satisfies
\begin{equation} \label{eq:scoring}
	\mathbb{E}_{\pi}\{\myS(\pi, Y)\} \leq \mathbb{E}_{\pi}\{\myS(p, Y)\}
\end{equation}
for all $p, \pi \in [0,1]$, where $\mathbb{E}_{\pi}(\cdot)$ denotes the expected value under the assumption that $Y = 1$ with probability $\pi$. That is, the true event probability attains a minimal expected score, and $\myS$ is strictly consistent if equality in \eqref{eq:scoring} only holds for $p = \pi$. Well-known examples are the Brier score $(y-p)^2$ and the logarithmic score $-\!\log(|1-y-p|)$.

For comparing the predictions $p_t$ and $q_t$, the forecast user would therefore collect a sample $(y_{t + h}, p_t, q_t)$, $t = 1, \dots, T$, and compute the empirical score difference $\frac{1}{T}\sum_{t = 1}^T \{\myS(p_t, y_{t + h}) -\myS(q_t, y_{t + h})\}$. To take into account the sampling uncertainty, such score differences are accompanied with p-values indicating whether the mean score significantly differs from zero. If the observations are not independent, as usual in sequential settings, a number of asymptotic tests are available to compute p-values, with prominent ones being the Diebold-Mariano test \citep{Diebold1995} and the test of conditional predictive ability by Giacomini and White \citep{Giacomini2006}. Further examples are the martingale-based approaches by \cite{Seillier1993} or \cite{Lai2011}, and more recent tests of forecast dominance \citep{Ehm2018, Yen2021}.

In this article, we expand the tools for drawing inference on probability forecast performance by e-values. E-values, with `e' referring to `expectation', have been introduced as an alternative to p-values for testing. The term e-value was used first in the literature by \citet{Vovk2021}, but the concept also appears in \cite{Shafer2021}, under the name `betting score', and in \cite{Gruenwald2020}; see also the series of working papers on \url{http://alrw.net/e/}. In brief, an e-value is a random variable $E \geq 0$ satisfying $\mathbb{E}(E) \leq 1$ under a given null hypothesis. By Markov's inequality, this implies $\mathbb{P}(E > 1/\alpha) \leq \alpha$ for any $\alpha \in (0,1)$, i.e.~large realizations of an e-value can be considered as evidence against the null hypothesis, and the value $1/E$ is a conservative p-value. A main motivation for using e-values instead of p-values, explained in more detail in \cite{Shafer2021}, \cite{Gruenwald2020} and \cite{Wang2020}, is their simple behaviour under combinations. The arithmetic average of e-values is again an e-value, and so is the product of independent or sequential e-values. E-values also have advantages over p-values with respect to false discovery rate control \citep{Wang2020}, which may be beneficial for the comparison of forecasts over many locations such as a fine latitute-longitude grid around the globe. The central property for this article is that e-values are valid under optional stopping and continuation, that is, the collection of data for computing an e-value may be stopped or continued based on seeing the past observations and e-values. It is well known that p-values in general do not satisfy these properties.

Our main contribution is the result that for any scoring rule $\myS$ and forecasts $p, q$ for $Y \in \{0,1\}$, there exists an e-value which satisfies $\mathbb{E}_{\pi}(E) \leq 1$ if and only if $\mathbb{E}_{\pi}\{\myS(p, Y) - \myS(q, Y)\} \leq 0$. This e-value allows one to draw inference on the relative performance of the forecasts $p$ and $q$ with respect to $\myS$ with only a single observation. In a sequential setting, e-values from different time points can be merged by products into a non-negative supermartingale or test-martingale, which are analysed in detail by \citet{Ramdas2020}. This gives a statistical test of forecast dominance which is valid in finite samples without any further assumptions on the data generating process. Moreover, the constructed e-values are valid under optional stopping, so a forecast user may decide to continue or stop forecast comparison based on only a part of the data. These advantages are inherent to any e-value, but we believe that they make e-values a particularly attractive tool in sequential forecast evaluation. The above mentioned tests for comparing probability forecasts are all only asymptotically valid, and the underlying assumptions are often difficult or impossible to verify. In the case of tests with asymptotic normality, the selection of the variance estimator for the test statistic may have a dramatic impact on the test validity \citep[see for example][Table 1]{Lazarus2018}. More serious is the problem of optional stopping. In a simple but realistic simulation example in this article, we demonstrate that commonly used tests for forecast superiority at the level of $0.05$ may yield rejection rates of up to $0.15$ under optional stopping, grossly misleading and invalidating statistical inference. Although statisticians and practitioners should know that the sample size for classical tests must be determined in advance, we believe that optional stopping is quite common in forecast evaluation, where data arrives sequentially and it might be tempting to stop, or continue, an expensive or time consuming experiment upon seeing enough, or just not enough, evidence against a hypothesis. Moreover, also in the analysis of past datasets, optional continuation may occur implicitly, in that methods are often first evaluated on a smaller, manageable part of the data and the analysis is continued if the results are promising. Last but not least, even to a statistician fully aware of the problem of optional stopping, it may be desirable to have a tool that allows stopping an evaluation when enough evidence is collected, without having to bother about the implications for inference.

The advantages of e-values for forecast comparison relative to the currently available methods come at a price, namely, lower power. This is well known not only for e-values, but a general phenomenon when tools for anytime-valid inference are compared to methods for inference with a fixed sample size; see for example Figure 1 in \cite{Waudabysmith2021} displaying the widths of time uniform and fixed time confidence intervals for a mean. However, in the case study in this article, p-values from classical tests and e-values lead to qualitatively the same results.

\section{Preliminaries}\label{sec:preliminaries}

\subsection{Scoring functions for probabilities} \label{sec:scoring}
Throughout the article, $\mathbb{E}_{\mathbb{Q}}(\cdot)$ denotes the expected value of the quantity in parentheses under the probability distribution $\mathbb{Q}$. If the measure $\mathbb{Q}$ is the probability $\pi \in [0,1]$ of a binary event, we simply write $\mathbb{E}_{\pi}(\cdot)$.

When comparing probability forecasts with scoring functions, the choice of the scoring function plays a crucial role. While \eqref{eq:scoring} guarantees that the true event probability always achieves a minimal expected score, different scoring functions may yield different rankings when misspecified forecasts are compared  \citep{Patton2020}. This problem can be avoided by basing forecast comparison on several or all scoring rules simultaneously. For probabilities of binary events, under mild regularity conditions stated in \citet[Theorem 2.3]{Gneiting2007}, all consistent scoring functions are of the form
\begin{equation} \label{eq:mixture}
	\myS(p, y) = \int_{(0,1)} \myS_{\theta}(p,y) \, d\nu(\theta),
\end{equation}
where $\nu$ is a locally finite Borel measure on $(0,1)$ and 
\begin{align} \label{eq:elementary}
	\myS_{\theta}(p, y) = (\theta - y)\{\one(p > \theta) - \one(y > \theta)\}
	= \begin{cases}
		\theta, & \quad y = 0, \ p > \theta, \\
		1 - \theta, & \quad y = 1,\  p \leq \theta, \\
		0, & \quad \text{otherwise}.
	\end{cases}
\end{align}
In the equation above, $\one$ denotes the indicator function. This representation originally dates back to \cite{Schervish1989}; see also \citet{Ehm2016}. The scoring function $\myS$ is strictly consistent if and only if $\nu$ assigns positive mass to all non-degenerate intervals in $(0,1)$.

\subsection{Forecast dominance and hypotheses} \label{sec:dominance}
Let $(\Omega, \mathcal{F}, \mathbb{Q})$ be a probability space with a filtration $\mathcal{F}_t$, $t \in \mathbb{N}$. We assume that the competing forecasts $p_t, q_t$ and the observation $Y_t$ are a random vector $(Y_t, p_t, q_t)$ adapted to $\mathcal{F}_t$, and $(p_t, q_t)$ are forecasts for $Y_{t + h}$ for some integer lag $h \geq 1$. The measure $\mathbb{Q}$ describes the joint dynamics of the forecasts and the observations.

When comparing forecasts using a given scoring function $\myS$, the quantity of interest is often not the unconditional expected score difference $\mathbb{E}_{\mathbb{Q}}\{\myS(p_t, Y_{t + h}) - \myS(q_t, Y_{t + h})\}$, which describes the average relative performance of $p_t$ and $q_t$. More interesting is the question whether given the information at the time of forecasting, $\mathcal{F}_t$, the conditional event probability is closer to $p_t$ than to $q_t$, i.e.~$\mathbb{E}_{\mathbb{Q}}\{\myS(p_t, Y_{t + h}) - \myS(q_t, Y_{t + h}) \mid \mathcal{F}_t\} \leq 0$. This notion of forecast dominance is called conditional forecast dominance and has been introduced by \cite{Giacomini2006}.

The definition of forecast dominance used here does not require knowledge about the processes generating $(Y_t, p_t, q_t)$, which are often unknown or not well enough understood to formulate a suitable stochastic model. The relative performance of the forecasts $p_t, q_t$ is governed by the underlying distribution $\mathbb{Q}$, and hypotheses about forecast dominance are hypotheses about the data generating process. Denoting by $\mathcal{P}$ the set of probability measures on $(\Omega, \mathcal{F})$, we will construct tests for the following hypotheses:
\begin{align}	
	\mathcal{H}_{\myS; c} & = \left[ \mathbb{P} \in \mathcal{P}: 
	c_t\mathbb{E}_{\mathbb{P}}\{\myS(p_t, Y_{t + h}) - \myS(q_t, Y_{t + h}) \mid \mathcal{F}_{t}\} \leq 0 \text{ a.s.}, \, t \in \mathbb{N}\right] \label{eq:h0s} \\	
	\mathcal{H}_c & = \left[\mathbb{P} \in \mathcal{P}: 
	\sup_{\theta \in [0,1]} c_t\mathbb{E}_{\mathbb{P}}\{\myS_{\theta}(p_t, Y_{t + h}) - \myS_{\theta}(q_t, Y_{t + h}) \mid \mathcal{F}_{t}\} \leq 0 \text{ a.s.}, t \in \mathbb{N}\right] \label{eq:h0}
\end{align}
Here, $(c_t)_{t \in \mathbb{N}}$ is a sequence of $\mathcal{F}_t$-measurable random variables $c_t \in \{0, 1\}$. If $c_t = 1$ for all $t$, we write $\mathcal{H}_{\myS; c} = \mathcal{H}_{\myS}$ and $\mathcal{H}_{c} = \mathcal{H}$. In this case, Hypothesis \eqref{eq:h0s} states that at all times $t$, forecast $p_t$ is at least as good as $q_t$ under the scoring rule $\myS$, given the information available at the time of forecasting. Hypothesis \eqref{eq:h0} is stronger and states that $p_t$ is preferred over $q_t$ under all elementary scores \eqref{eq:elementary}, and it corresponds to what is denoted by $H^s_{-}$ in \citet[formula (2.5)]{Ehm2018}. Recently, hypotheses of the type of $\mathcal{H}$ or $\mathcal{H}_{\myS}$ have been put in question by \cite{Zhu2020}, who demonstrate that the null hypothesis of equal conditional predictive accuracy is basically never satisfied in realistic settings. Their criticism does not directly apply to one-sided hypotheses, but we emphasize that the null hypotheses $\mathcal{H}_{\myS}$ or $\mathcal{H}$ are rather strong in that they require conditional dominance at all time points. Tests for these hypotheses are therefore most suitable for the comparison of a new method to an established benchmark or state-of-the-art method, where rejecting the null means that the new method outperforms the benchmark at least in some situations -- a minimal requirement. 

The classical example for a situation with $\mathbb{P} \in \mathcal{H}$ is $p_t = \mathbb{P}(Y_{t+h} = 1 \mid \mathcal{F}_{t})$, i.e.~$p_t$ is the ideal forecast in the sense of \cite{Gneiting2013}. For the hypotheses $\mathcal{H}_{\myS}$, one may easily construct situations with dominance relations also among non-calibrated forecasts; see the simulation examples in Section \ref{sec:simulations}.

In many practical situations, it cannot be expected that a forecast method always outperforms another one, and forecast users are rather interested in the question under what conditions a particular forecast should be preferred. Choosing the sequence $(c_t)_{t \in \mathbb{N}}$ such that $c_t = 1$ if the condition holds and $c_t = 0$ otherwise allows to formalize this question. Here the variables $c_t$ must be $\mathcal{F}_t$-measurable, that is, known at the time of forecasting. In practice this is not a severe limitation, since the information that one forecast is more accurate than another one under a given condition is only useful if this condition is known at the time of forecasting, and not ex post. But also from a theoretical point of view, forecast evaluation should only be conditioned on the forecasts themselves, and not on the observations or on information not available at the time of forecasting; see \cite{Lerch2017} for a detailed analysis of this issue in the case of extreme events.

\section{E-values for testing forecast dominance} \label{sec:main}

\subsection{One-period setting} \label{sec:oneperiod}
We first construct e-values for the comparison of probability forecasts in a one-period setting where $Y = 1$ with probability $\pi$ and the forecasts $p,q$ are assumed to be fixed numbers in $(0,1)$. These e-values give an absolute and valid interpretation of predictive performance with only a single observation, e.g.~for a single time point in the sequential setting of Section \ref{sec:dominance}, or in binary classification problems with independent forecast-observation pairs, where the competing forecasts are based on covariates and $\pi$ is the probability that $Y = 1$ conditional on the covariate values. The null hypotheses that $p$ is a better forecast than $q$ with respect to a given score $\myS$, or with respect to all scoring functions simultaneously, here correspond to
\begin{align*}
	H_{\myS} & = [\pi \in [0,1]: \, \mathbb{E}_{\pi}\{\myS(p, Y) - \myS(q, Y)\} \leq 0], \\
	H & = \left[\pi \in [0,1]: \, \sup_{\theta \in [0,1]} \mathbb{E}_{\pi}\{\myS_{\theta}(p, Y) - \myS_{\theta}(q, Y)\} \leq 0 \right].
\end{align*}
For $p < q$, a direct computation shows that $H_{\myS}$ is the interval $[0, \kappa_{\nu}\{[p, q)\}]$ with
\[
\kappa_{\nu}\{[a,b)\} = \frac{\int_{[a,b)} \theta \, d\nu(\theta)}{\nu\{[a,b)\}}, \ 0 < a < b < 1.
\]
The stronger null hypothesis $H$ is the intersection of these intervals for all mixing measures $\nu$, that is, $[0, p]$. In the case $q > p$, the intervals take the form $[\kappa_{\nu}\{[q, p)\}, 1]$ or $[p, 1]$, respectively.

For a set $\mathcal{P}$ of probability measures and disjoint $H, H' \subset \mathcal{P}$, we say that an e-value $E$ has null hypothesis $H$ and alternative $H'$ if $\mathbb{E}_{\mathbb{P}}(E) \leq 1$ for all $\mathbb{P} \in H$ and $\mathbb{E}_{\mathbb{Q}}(E) > 1$ for all $\mathbb{Q} \in H'$. The following theorem characterizes e-values for testing $H_{\myS}$.

\begin{theorem} \label{thm:uniqueness}
	Let $\myS$ be a consistent scoring function and $p, q \in (0,1)$, $p \neq q$. Assume that the mixing measure $\nu$ of $\myS$ satisfies $\nu\big\{[\min(p,q), \max(p,q))\big\} > 0$. Then a function $E = E(y)$ is an e-value with null hypothesis $H_{\myS}$ and alternative $[0,1] \setminus H_{\myS}$, if and only if for some $\lambda \in (0,1]$,
	\begin{equation} \label{eq:representation}
		E(y) = E_{p,q; \lambda}(y) = 1 + \lambda \frac{\myS(p, y) - \myS(q, y)}{|\myS(p, \one\{p > q\}) - \myS(q, \one\{p > q\})|}.
	\end{equation}
\end{theorem}

Theorem \ref{thm:uniqueness} gives a family of e-values for testing forecast dominance with a given score $\myS$, and in a next step, we tune the parameter $\lambda$ in \eqref{eq:representation} such that the corresponding e-value has maximal `power' against a given alternative. The notion of power for e-values differs from the classical power of p-values, and it is motivated in detail by \cite{Shafer2021} and \cite{Gruenwald2020}. An e-value can be interpreted as a bet against the null hypothesis, and a product $\prod_{t = 1}^T E_t$ of e-values represents the accumulated capital at time if $T$ the initial capital is $1$ and all money is invested in the bet at each step. Maximizing the gains is equivalent to maximizing the `growth rate' $(1/T)\log \prod_{t = 1}^T E_t = (1/T)\sum_{t = 1}^T \log (E_t)$, a strategy which is sometimes called Kelly betting, in reference to \citet{Kelly1956}. If an e-value maximizes $\mathbb{E}_{\mathbb{P}}\{\log(E)\}$ under a measure $\mathbb{P}$ representing an alternative hypothesis, it is called growth rate optimal or simply GROW \citep{Gruenwald2020}. One such alternative could be that $Y = 1$ with probability $q$, but one can maximize the power under any other alternative $\pi_1 \not\in H_{\myS}$.

\begin{theorem} \label{thm:grow}
	Under the assumptions of Theorem \ref{thm:uniqueness}, for any $\pi_1 \not\in H_{\myS}$, $\mathbb{E}_{\pi_1}\{\log(E_{p,q;\lambda})\}$ is maximal in $\lambda$ if and only if
	\[
	\lambda = \begin{dcases}
		(1-\pi_1) + \pi_1 \frac{\myS(p, 1) - \myS(q, 1)}{\myS(p, 0) - \myS(q, 0)}, & \ p > q, \\
		\pi_1 + (1-\pi_1) \frac{\myS(p, 0) - \myS(q, 0)}{\myS(p, 1) - \myS(q, 1)}, & \ p < q.
	\end{dcases}
	\]
	The corresponding e-value equals
	\[
	E_{p,q}^{\pi_1}(y) = \begin{dcases}
		\frac{1 - \pi_1}{1 - \kappa_{\nu}\{[\min(p,q),\max(p,q))\}}, & \ y = 0, \\
		\frac{\pi_1}{\kappa_{\nu}\{[\min(p,q),\max(p,q))\big\}}, & \ y = 1.
	\end{dcases}
	\]
\end{theorem}

Theorem \ref{thm:grow} shows that the GROW e-values for the comparison of probability forecasts take the form of likelihood ratios with the alternative probability in the numerator and the integral of the mixing measure $\nu$ (suitably normalized) over the interval $[\min(p,q), \max(p,q))$ in the denominator. It is possible to obtain this result directly by applying Theorem 1 in \cite{Gruenwald2020}, since $\kappa_{\nu}\{[\min(p,q),\max(p,q))\}$ is the boundary of the null-hypothesis $H_{\myS}$. We chose the indirect but more instructive approach via Theorem \ref{thm:uniqueness} since to the best of our knowledge, this is the first application of e-values to forecast comparison, and similar approaches might be used to construct e-values for score differences in more general settings than the evaluation of binary event forecasts. In fact,  \citet[Proposition 2]{Waudabysmith2021} have a similar representation as in \eqref{eq:representation} for e-values for testing hypotheses about a constant mean.

\begin{table}[t]
	\centering
	\caption{Commonly used scoring rules and the corresponding denominator in the GROW e-value under the assumption $p < q$. The case $p > q$ is obtained by interchanging the roles of $p$ and $q$. The mixing measure $\nu$ is given in the form of its Lebesgue density $h(\theta)$, $\theta \in (0,1)$. For the spherical score, $\|p\| := (2p^2-2p+1)^{1/2}$ denotes the Euclidean norm of the vector $(p, 1-p)$. \label{tab:scores}}
	\bigskip
	\begin{tabular}{cccc}
		Score & $\myS(p,y)$ & Mixing density $\nu$ & $\kappa_{\nu}\{[p,q)\}$ \\[0.75em]
		Brier & $(p-y)^2 $ & $2$ & $(p + q)/2$ \\[0.5em]
		Logarithmic & $-\log(|1-y-p|)$ & $\theta^{-1}(1-\theta)^{-1}$ & $\log\left(\frac{1-p}{1-q}\right) \Big/\log\left(\frac{q(1-p)}{p(1-q)}\right)$ \\[0.5em]
		Spherical & $1 - |1-y-p|/\|p\|$ & $(2\theta^2 -2\theta + 1)^{-3/2}$ & $\frac{(q-1)\|p\| - (p-1)\|q\|}{(2q-1)\|p\|-(2p-1)\|q\|}$ \\
	\end{tabular}
\end{table}

For the test of the null hypothesis $H$, applying \citet[Theorem 1]{Gruenwald2020} shows that the GROW e-value is the likelihood ratio.

\begin{theorem} \label{thm:allscores}
	Let $p,q \in (0,1)$. Then the GROW e-value with null hypothesis $H$ and alternative hypothesis that $Y = 1$ with probability $\pi_1 \not\in H$ is given by
	\[
	E^{\pi_1*}_{p,q}(y) = \begin{dcases}
		(1-\pi_1)/(1-p), & \ y = 0,\\
		\pi_1/p, & \ y = 1.\\
	\end{dcases}
	\]
\end{theorem}

In testing with e-values, the GROW e-value for testing the point null hypothesis $\{p\}$ against the alternative $\pi_1$ is exactly the likelihood ratio, and Theorem \ref{thm:allscores} states that this is equivalent to testing forecast dominance with respect to all scoring functions. Dominance with respect to all scoring functions is a very strong requirement on $p$, since the null hypothesis is false as soon as the true probability $\pi$ is on the same side of $p$ as $q$, that is, in $(p, 1]$ for $p < q$ or in $[0, p)$ for $q < p$, and the choice of $\pi_1$ is restricted to these sets. Unlike the e-values $E_{p,q}^{\pi_1}$, $E^{\pi_1*}_{p,q}$ does not depend directly on $q$, but indirectly via the admissible values for $\pi_1$.

\subsection{Sequential inference} \label{sec:sequential}
We now turn to the sequential model with observations $Y_t$ and forecasts $p_t, q_t$ defined on a probability space $(\Omega, \mathcal{F}, \mathbb{Q})$ with a filtration $\mathcal{F}_t$, $t \in \mathbb{N}$. 
In the case $h = 1$, for any $\mathbb{Q} \in \mathcal{H}_{\myS; c}$ and any adapted sequence $\lambda_t \in [0,1]$, $t \in \mathbb{N}$, with $E_{p_t, q_t; \lambda_t}$ as defined in \eqref{eq:representation},
\begin{align*}
	\mathbb{E}_{\mathbb{Q}} \left\{ \prod_{t = 1}^T E_{p_t, q_t; \lambda_t}(Y_{t + 1}) \right\} & \ = \ \mathbb{E}_{\mathbb{Q}} \left[\mathbb{E}_{\mathbb{Q}} \left\{ \prod_{t = 1}^T E_{p_t, q_t; \lambda_t}(Y_{t + 1}) \mid \mathcal{F}_{T}\right\}\right] \\
	& \ = \ \mathbb{E}_{\mathbb{Q}} \left[ \prod_{t = 1}^{T-1}E_{p_t, q_t; \lambda_t}(Y_{t+1})\mathbb{E}_{\mathbb{Q}} \left\{ E_{p_T, q_T; \lambda_T}(Y_{T+1}) \mid \mathcal{F}_{T}\right\} \right].
\end{align*}
If $c_t = 0$, then there is no hypothesis about $p_t$ and $q_t$. For these cases, the definition at \eqref{eq:representation} may be extended to $\lambda = 0$, so that $E_{p_t, q_t; 0} \equiv 1$ if $c_t = 0$. Then, if $\lambda_T = 0$ when $c_T = 0$,
\[
\mathbb{E}_{\mathbb{Q}} \left\{ E_{p_T, q_T; \lambda_T}(Y_{T+1})\mid \mathcal{F}_T\right\} = (1 - c_T) + c_T \mathbb{E}_{\mathbb{Q}} \left\{ E_{p_T, q_T; \lambda_T}(Y_{T+1})\mid \mathcal{F}_T\right\} \leq 1
\]
almost surely for $\mathbb{Q} \in \mathcal{H}_{\myS; c}$, so
\[
\mathbb{E}_{\mathbb{Q}} \left\{ \prod_{t = 1}^T E_{p_t, q_t; \lambda_t}(Y_{t + 1}) \right\}\leq \ \mathbb{E}_{\mathbb{Q}} \left\{ \prod_{t = 1}^{T-1}E_{p_t, q_t; \lambda_t}(Y_{t+1}) \right\}.
\]

Iterating this argument shows that the product $\prod_{t = 1}^T E_{p_t, q_t; \lambda_t}(Y_{t + 1})$ is an e-value for $\mathcal{H}_{\myS; c}$; more precisely, the process $\prod_{j = 1}^t E_{p_j, q_j; \lambda_j}(Y_{j + 1})$, $t = 2, 3, \dots$, is a non-negative supermartingale with respect to $(\mathcal{F}_t)_{t \in \mathbb{N}}$. For general lag $h$, sequential conditioning at time steps of $1$ is not possible, and one option is to average the products of all e-values with time difference of $h$, in the spirit of the U-statistics merging functions suggested by \citet{Vovk2021}. We summarize this in the following proposition.

\begin{proposition} \label{prop:combination}
	Let $(Y_t, c_t, p_t, q_t, \lambda_t) \in \{0, 1\}^2 \times (0,1)^2 \times [0,1]$ be defined on a measurable space $(\Omega, \mathcal{F})$ and adapted to the filtration $\mathcal{F}_t$, $t \in \mathbb{N}$, and assume $\lambda_t = 0$ if $c_t = 0$. Let further $\myS$ be a strictly consistent scoring function. Then for all $T \geq h + 1$, with $I_k = \{k + hs: \, s = 0, \dots, \lfloor (T-k)/h\rfloor - 1\}$, \[
	e_T = \frac{1}{h}\sum_{k = 1}^{h} \prod_{l \in I_k}E_{p_l, q_l; \lambda_l}(Y_{l + h})
	\]
	are $\mathcal{F}_T$-measurable and are e-values under $\mathcal{H}_{\myS; c}$.
\end{proposition}

Proposition \ref{prop:combination} is an analogous result to Theorem \ref{thm:uniqueness} in the sense that it only characterizes possible e-values for testing forecast dominance, but the parameters $\lambda_t$ could be any adapted sequence $(\lambda_t)_{t \in \mathbb{N}}  \subset [0,1]$. E-values for dominance testing under the conditions $(c_t)_{t \in \mathbb{N}}$ are obtained by setting all e-values where the condition is not satisfied to $1$. The forecast user may, and in fact, has to, tune the $(\lambda_t)_{t \in \mathbb{N}}$ in order to attain a good power against a given alternative. Recall that at any $t$, the $\lambda_t$ may be a function of all the forecasts and observations before time $t$. Instead of the parameters $\lambda_t$, it is usually more intuitive to think about an alternative probability $\eta_t$ for the event $Y_{t + h} = 1$, and then directly use the GROW e-values $E^{\eta_t}_{p_t, q_t}$ constructed in Theorem \ref{thm:grow}. In that respect, testing forecast dominance with e-values differs from p-value based tests of a zero score difference, which do not require the user to explicitly specify an alternative hypothesis. In the applications in Sections \ref{sec:simulations} and \ref{sec:case}, we will give guidance on the selection of alternative hypotheses, and show that reasonable power may be attained with simple heuristic methods.

As a side remark, choosing an alternative hypothesis for e-values in sequential forecast dominance testing is similar to the conditional predictive ability tests by \cite{Giacomini2006}, where $\mathcal{F}_t$-measurable test functions are used to weight score differences and improve power. While the selection of the test functions in the Giacomini-White test is delicate, because they may have an impact on the variance estimates and the finite-sample validity of the tests, e-values remain valid under any choice of adapted weights $(\lambda_t)_{t \in \mathbb{N}}$.

Our last theoretical result states that the e-values $e_T$ constructed above are also valid when $T$ is replaced by a stopping time $\tau$. For $h = 1$, this is a consequence of the fact that $(e_t)_{t \geq 2}$ is a non-negative supermartingale \citep[see][Section 3]{Ramdas2020}.

\begin{proposition} \label{prop:stopping}
	Let $\tau \in \mathbb{N}$ be a stopping time. Then under the assumptions of Proposition \ref{prop:combination},
	\[
	\mathbb{E}_{\mathbb{Q}}(e_{\tau+h-1}) \leq 1, \ \mathbb{Q} \in \mathcal{H}_{\myS}.
	\]
\end{proposition}

To understand validity under optional stopping intuitively, recall that at time $t$, the forecast user has to determine the parameter $\lambda_{t}$ in the e-value $E_{p_t, q_t; \lambda_t}(Y_{t + h})$. Optional stopping at $t_0$ corresponds to setting $\lambda_t \equiv 0$, or equivalently $E_{p_t, q_t; \lambda_t}(Y_{t + h}) \equiv 1$, for $t \geq t_0$, i.e.~ignoring all observations starting from time $t_0 + h$. In the case $h = 1$, this allows the forecast user to stop evaluation at any time, since $\lambda_{t}$ in $E_{p_t, q_t; \lambda_t}(Y_{t+1})$ is defined at the same time as $Y_{t}$ is observed. However, when $h > 1$, the coefficients $\lambda_t$ in $E_{p_t, q_t; \lambda_t}(Y_{t + h})$ for $t = t_0 - h + 1, \dots, t_0 - 1$ have already been determined in the past and may not be set to zero at $t_0$, since they must be $(\mathcal{F}_t)_{t \in \mathbb{N}}$-adapted. This implies that the stopped e-value depends on the unknown, future observations $Y_{t_0 + 1}, \dots, Y_{t_0 + h - 1}$, so it is not deterministic at time $t_0$.

In the case $h = 1$, optional stopping is a powerful strategy when the goal is to assess forecast superiority at a significance level $\alpha \in (0,1)$, since the stopping time
\begin{equation*}
	\tau_{\alpha} = \min\{T, \, \inf(t \geq 2: e_t \geq 1/\alpha)\}
\end{equation*}
allows to reject the null hypothesis as soon as the sequential e-value $e_t$ exceeds $1/\alpha$. If $h > 1$, one may similarly define
\[
\tau_{\alpha, h} = \min\Big(T, \, \inf\big[t \geq h + 1: e_t \geq \max_{j = t - h + 1, \dots, t - 1} \!\!\!\! E_{p_{j}, q_{j}; \lambda_j}\{\one(p_{j} > q_{j})\}^{-1}  /\alpha\big]\Big),
\]
which guarantees that when stopping at $t_0$, the level $1/\alpha$ is exceeded no matter what values $Y_{t_0 + 1}, \ldots, Y_{t_0 + h - 1}$ take; see Appendix \ref{app:stopping}. Instead of specifying a significance level $\alpha$ in advance, one may as well transform the sequence $(e_t)_{t \in \mathbb{N}}$ into so-called anytime valid p-values, which are valid simultaneously for all $t_0 \geq h + 1$ \citep[see][Section 3.1]{Ramdas2020}. For $h = 1$, an anytime-valid p-value is given by $p_{t_0} = \min\{1, \inf_{s = 1, \dots, t_0} 1/e_{s}\}$. For $h > 1$, one may apply the same correction as in the stopping time $\tau_{\alpha,h}$, namely, define $p_{t_0} = \min(1, \inf_{s = 1, \dots, t_0} [\max_{j = s - h + 1, \dots, s - 1} \! E_{p_{j}, q_{j}; \lambda_j}\{\one(p_{j} > q_{j})\}^{-1}/e_{s}])$

\section{Simulation examples} \label{sec:simulations}

\subsection{Basic properties} \label{sec:basics}
For the simulation examples in this and in the next subsection, we will transform e-values $E$ into p-values by taking their inverse $1/E$, so that direct comparisons with p-values are possible. Further variations of these simulation examples are presented in \ref{app:simulation}. An \textsf{R} package for the proposed methods and replication material for all results in this article are available on GitHub (\url{https://github.com/AlexanderHenzi/eprob}).

In the first example, for varying $\mu \in (0,1)$, we simulate independent forecasts $p_t, \, q_t \sim \mathrm{Unif}(0,1)$, define $\pi_t = \mu q_t + (1-\mu)p_t$, and generate independent Bernoulli observations $Y_{t + 1}$ with mean $\pi_t$ conditional on $p_t, q_t$. This represents a situation where forecasters only have access to partial information and both forecasts are not calibrated, i.e.~$\mathbb{P}(Y_{t+1} = 1 \mid p_t) \neq p_t$ and $\mathbb{P}(Y_{t+1} = 1 \mid q_t) \neq q_t$. We choose $\myS$ to be the Brier score, so that $p_t$ outperforms $q_t$ if and only if $\pi_t \in [0, (p_t + q_t)/2]$ if $p_t < q_t$ or $\pi_t \in [(p_t + q_t)/2, 1]$ if $p_t > q_t$, i.e.~if and only if $\mu \leq 0.5$. When $\mu > 0.5$, the GROW e-value is obtained by choosing $\pi_t$ as alternative hypothesis probability, but in practice, $\pi_t$ is not known. The forecast user might assume that the true probability of $Y_{t + 1} = 1$ lies somewhere in between $(p_t + q_t)/2$ and $q_t$, and choose a convex mixture $\eta_t(\xi) = \xi (p_t + q_t)/2 + (1-\xi)q_t$ with some $\xi \in (0,1)$ as alternative. Proposition \ref{prop:combination} implies that for $k \in \mathbb{N}$ and $\xi_1, \dots, \xi_k \in (0,1)$,
\[
e_{t; \xi_j} = \prod_{i = 1}^t E^{\eta_i(\xi_j)}_{p_i, q_i}(Y_{i + 1}), \quad e_t = \frac{1}{k}\sum_{j = 1}^k e_{t;\xi_j}
\]
are e-values under $\mathcal{H}_{\myS}$. In Figure \ref{fig:sim1}, we compare the rejection rates at the 5\% level, corresponding to e-values greater or equal to 20, when the $\xi_j$ are $k$ equispaced weights in $(0,1)$ for $k = 1$ and $k = 5$, i.e.~$\xi_1 = 0.5$ if $k = 1$ and $\xi_l = l/6$, $l = 1, \dots, 5$, in the case $k = 5$. We computed both the unstopped e-value $e_T$ and the stopped variant $e_{\tau_{0.05}}$, and the e-values under alternatives $\eta_t = \pi_t$ and $\eta_t = q_t$. The rejection rates are compared to those of one-sided t-tests of the null hypothesis that the mean Brier score difference is non-positive. Additionally, we show the rejection rates when the p-value is used for optional stopping at given time points upon seeing a significant difference.

\begin{figure}[ht]
	\includegraphics[width=0.9\textwidth]{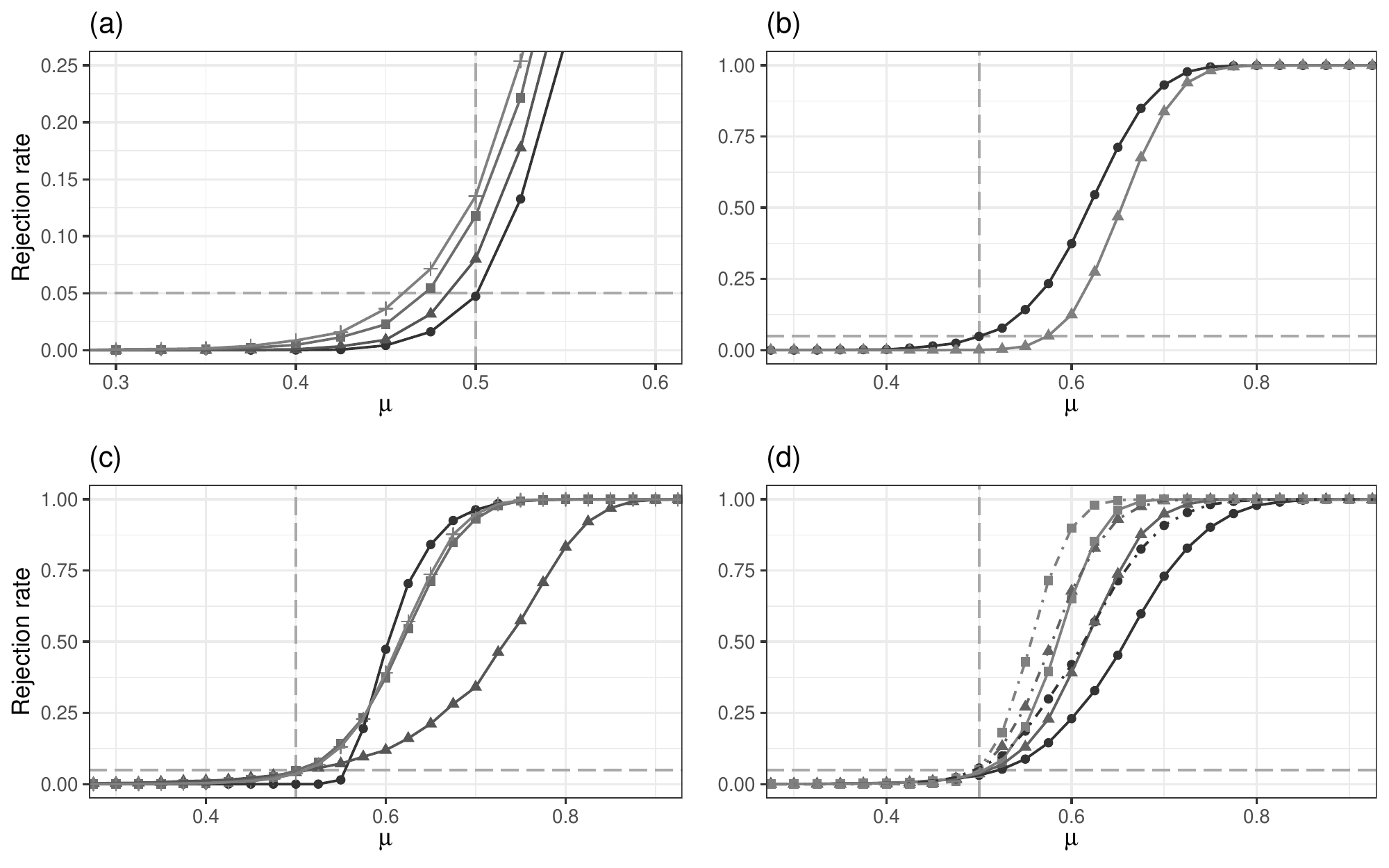}
	\centering
	\caption{Rejection rates of e-values and Student's t-test for the hypothesis that $p_t$ dominates $q_t$ with respect to the Brier score in the simulation of Section \ref{sec:basics}. Sample size is $T = 600$ for the panels (a)-(c) and the significance level is $\alpha = 0.05$ for all panels. (a) Rejection rate of t-test under optional stopping at $1, \, 3$, and $5$ equispaced time points (triangles, squares, crosses) in between $1$ and $T = 600$, and without optional stopping (dots). (b) Rejection rates of stopped (dots) and unstopped (triangles) e-value with $k = 1$. (c) Rejection rates of e-values with different alternative hypotheses ($q_t$: triangles, $\pi_t$: dots, $k = 1$: crosses, $k = 5$: squares). (d) Rejection rates of e-value ($k = 5$; normal lines) and t-test (without stopping; dot-dashed lines) for varying sample size $T$ (dots: $300$, triangles: $600$, squares: $1200$). \label{fig:sim1}}
\end{figure}

Our simulations illustrate the known fact that classical statistical tests are not valid under stopping. At the boundary of the null hypothesis, the rejection rate of the t-test amounts to $0.12$ for $T = 600$ and optional stops at times $150$, $300$ and $450$; given the number of optional stops, this phenomenon occurs independently of the sample size. As for the e-values, stopping ($e_{\tau_{0.05}}$) is always a more powerful but valid strategy compared to the e-value $e_T$. While the heuristic alternatives achieve a power close to the power under the correct alternative hypothesis, the misspecified hypothesis $\eta_t = q_t$ is clearly weaker. Interestingly, the correct alternative $\eta_t = \pi_t$ has a lower power than the heuristic alternatives close to the boundary of the null hypothesis. This is not an error: Specifying $\eta_t = \pi_t$ yields the optimal growth rate for the e-value, but this does not necessarily mean that it gives optimal power for the stopped e-value at the threshold $1/\alpha = 20$ in finite samples. The t-test generally achieves a higher power than the e-values, which has to be expected given the absence of assumptions on the data generating process and the validity under optional stopping; see also \cite{Waudabysmith2021}.

\subsection{Time series example} \label{sec:ts}

We simulate $Z_t$ from a moving average process $Z_t = \epsilon_t + \theta\sum_{j = 1}^4 \epsilon_{t-j}$, and define
\begin{equation} \label{eq:sim2}
	Y_t = \one\{Z_t > 0\}, \quad \pi_{t; h} = \mathbb{P}(Z_t > 0 \mid Z_{t - j}, \, j = h, \dots, 4), \ h = 1, \dots, 4.
\end{equation}
The probability $\pi_{t; h}$ corresponds to the ideal forecast at lag $h$. We compare $q_{t; h} = \pi_{t; h}$ and $p_{t; h} = \pi_{t; h + 1}$ for lags $h = 1, \dots, 3$, so that $q_{t; h}$ always outperforms $p_{t; h}$. With decreasing parameter $\theta$, serial dependence decreases and the forecast skill of $p_{t; h}$ and $q_{t; h}$ becomes similar. The alternative hypothesis for the e-values is the correct alternative $\eta_{t; h} = q_{t; h}$, so that the effect of a higher lag can be analyzed isolated from the question how to choose the alternative hypothesis. Rejection rates are compared to the Diebold-Mariano test at the 5\%-level.

\begin{figure}[t]
	\caption{Rejection rates of e-value (dots) and Diebold-Mariano test (triangles) in Example \eqref{eq:sim2} at the $5\%$ level for different sample sizes (panel rows) and lags $h$ (panel columns). \label{fig:sim2}}
	\bigskip
	\includegraphics[width=0.9\textwidth]{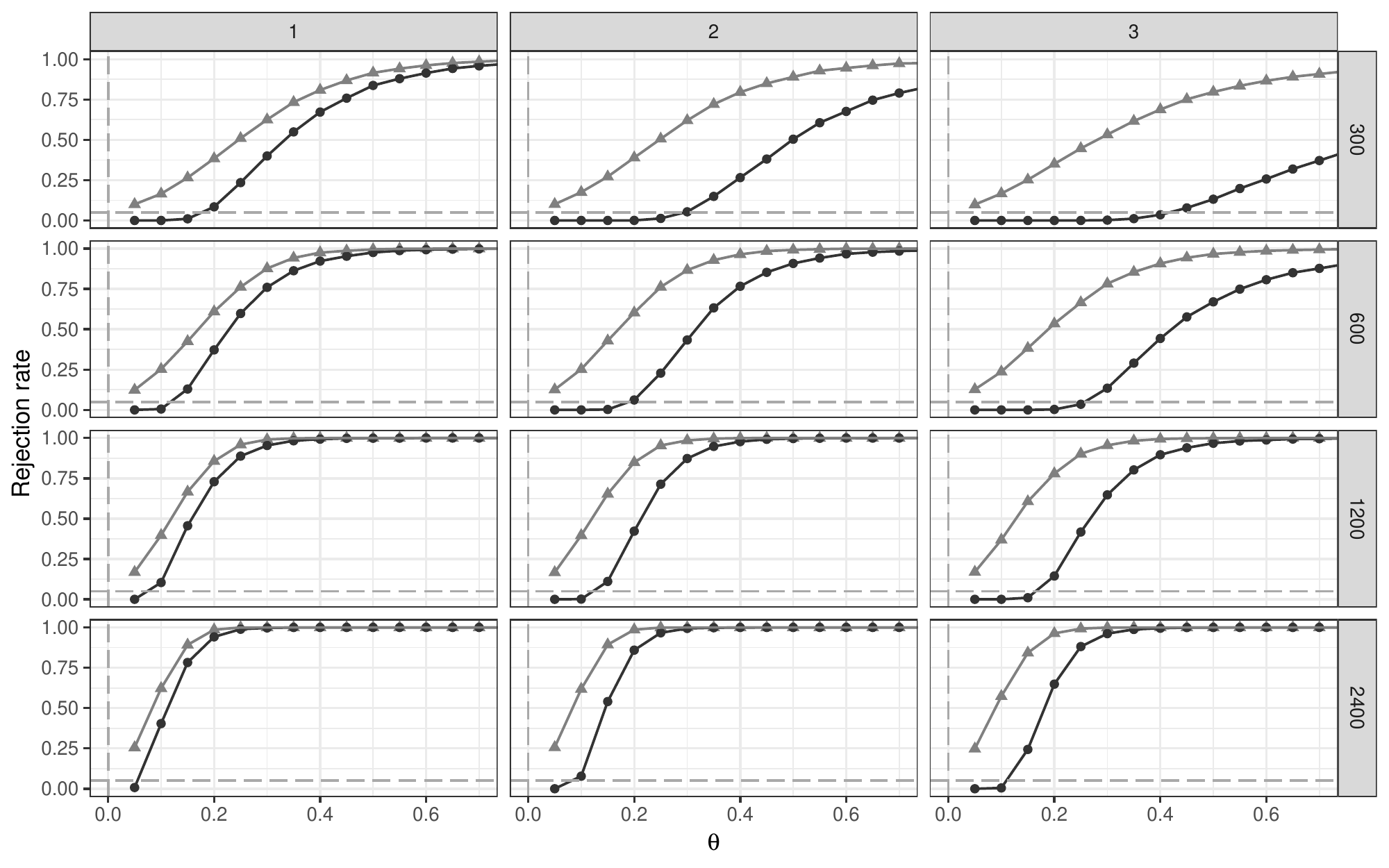}
\end{figure}

Figure \ref{fig:sim2} shows the rejection rates depending on the parameter $\theta$ for different sample sizes $T$. The e-values use the stopping time $\tau_{0.05}$ for lag $1$ and $\tau_{0.05; h}$ for lags $h = 2$ and $h = 3$. As in the previous simulations, the power of the e-values is below the p-values for the lag 1 forecasts, where the Diebold-Mariano test essentially coincides with the t-test. For lags $2$ and $3$, this difference increases, since the combination method for e-values becomes less powerful. With increasing lag, the rejection rates of both methods decrease, but the difference to lag 1 is smaller for the Diebold-Mariano test compared to the e-value. In this example, the Diebold-Mariano test is valid because the forecasts are ideal and the data generating process is stationary. For the e-values, validity is guaranteed without such assumptions, which may be of great advantage in applications.

\section{Case study} \label{sec:case}

\subsection{Data and methods}
\citet{Henzi2021} have compared postprocessing methods for precipitation 
forecasts with lag 1 to 5 days at the airports Brussels (BRU), Frankfurt (FRA), London Heathrow (LHR) and Zurich (ZRH). In their case study, probability of precipitation (PoP) forecasts have been evaluated with the Brier score, but no tests for significance of score differences have been performed. We will illustrate here how to apply e-values for probability forecasts, and compare the results to state-of-the-art forecast dominance tests.

A detailed description of the dataset and methods is given in Section 5 of
\citet{Henzi2021}, and we only summarise the key information here. The dataset covers the period from January 06, 2007 to January 01, 2017, and accounting for missing values, the numbers of available observations are 3406 for Brussels, 3617 for Frankfurt, 2256 for London and 3241 for Zurich airport. Postprocessing is applied to the ensemble forecasts of the European Centre for Medium-Range Weather Forecasts \citep[ECMWF;][]{Molteni1996, Buizza2005}, which are issued on a latitude-longitude grid and consist of a high resolution forecast, 50 perturbed ensemble forecasts at a lower resolution, and the control run for the perturbed forecasts. In simple words, ensemble forecasts account for uncertainty by running a numerical weather prediction (NWP) model several times, each time under slightly perturbated initial conditions, and each run of the model yields a different forecast, which together form a so-called ensemble \citep{Leutbecher2008}. Ensemble forecasts are usually subject to biases and dispersion errors, which can be corrected by estimating the conditional distribution of the weather variable given the NWP ensemble. This statistical procedure is known as postprocessing of ensemble forecasts \citep{Vannitsem2018}.

\citet{Henzi2021} propose isotonic distributional regression (IDR) as a benchmark for such postprocessing methods. IDR estimates conditional distributions nonparametrically and without any tuning parameters.
The method is not specifically tailored to forecasting precipitation, and one would expect that a parametric model designed for this purpose gives more precise forecasts. One such method is heteroscedastic censored logistic regression \citep[HCLR;][]{Messner2014}, which assumes that the square root of the precipitation follows a logistic distribution censored at zero. The implementation is as in \citet{Henzi2021}. While the covariates in IDR are only the high resolution forecast, the control forecast, and the ensemble mean, the HCLR model additionally includes a scale parameter depending on the ensemble standard deviation.

Different from the study in \cite{Henzi2021}, who use an expanding window for the postprocessing, we estimate both postprocessed forecasts on half of the data for each airport for simplicity, and keep the remaining half for validation.

\subsection{Hypothesis tests}
We illustrate the usage of e-values in the following hypothesis tests. Firstly, we try to reject the null hypothesis that IDR probability of precipitation forecasts are better than the HCLR PoP forecasts with respect to the Brier score. Secondly, we modify HCLR by dropping the scale parameter. It is expected that this variant, denoted by HCLR$_{-}$ subsequently, is outperformed by HCLR including the ensemble-dependent scale parameter, and also by IDR, since the models are based since both IDR and HCLR$_{-}$ assume a monotone relationship between the covariates and the PoP, but the nonparametric IDR can estimate a broader class of functions. And finally, we further investigate the effect of the scale parameter on HCLR predictions for high precipitation. Suppose a weather forecaster issues a warning if the probability that the precipitation exceeds a high threshold is more than $50\%$. As thresholds, we chose the empirical $90\%$ quantile of precipitation in the training data for each airport. Intuitively, the HCLR model should yield more accurate warnings than HCLR$_{-}$, because it includes the ensemble standard deviation as an uncertainty measure.

The first and second set of hypotheses are tested with the Brier score and the corresponding e-values. As alternative probability, we take the convex mixtures $\eta_t = 0.25p_t + 0.75 q_t$, which have been explored in Section \ref{sec:simulations}, denoting by $p_t$ the forecasting method that is expected to have a better performance than $q_t$ under the null hypothesis. The hypothesis about the extreme precipitation warnings is a conditional comparison with the conditions $c_t = \one\{\max(p_t, q_t) \geq 0.5\}$. For this hypothesis, instead of dominance with respect to the Brier score, we test the stronger hypothesis of forecast dominance with respect to all scoring rules. The rationale is that the forecast dominance hypothesis should be easily rejected if the HCLR model truly issues the better tail forecasts, and on the other hand, failing to reject may indicate that either even with data of 10 years it is not possible to clearly discriminate the quality of such warnings, or that the ensemble standard deviation does not bring a benefit. For this hypothesis we define $\eta_t = q_t$, assuming that the conditional event probabilities should be much closer to the ones issued by HCLR than by HCLR$_{-}$. No optional stopping is applied in all e-values.

For comparison, we also compute p-values for the significance of score differences. The first two hypotheses are tested with one-sided Diebold-Mariano tests (\citealp{Diebold1995}; see also \citealp{Giacomini2006}). To estimate the variance of the test statistics, we use the heteroskedasticity and autocorrelation consistent estimator with Bartlett weights, see \citet[Equation 2.18]{Lerch2017}. For testing dominance of the tail probability forecasts, the test by \citet{Yen2021} would allow arbitrary forecast lags, but it assumes strict stationarity. Since the sequence $c_t$ selects only particular instances, with possibly strongly varying time gaps in between, stationarity is highly questionable. We therefore apply the dominance test by \citet{Ehm2018}, which is valid under weaker assumptions but limited to lag $1$. Strictly speaking, both the Diebold-Mariano test and the forecast dominance test are valid under larger null hypotheses than the e-values, as they only require the average score difference between $p_t$ and $q_t$ to be non-positive, whereas the null hypothesis for the e-values asks for conditional superiority at each time point. A comparison is nevertheless interesting, since these two tests represent commonly used methods that are applied to test the significance of score differences.

\begin{table}
	\centering
	\caption{Brier scores for different probability of precipitation forecasting methods, and e-values ($E$) and p-values ($p$) for testing significance of score differences. The columns HCLR/IDR show e-values and p-values for tests of the null hypothesis that IDR probability of precipitation forecasts achieve a lower Brier score the HCLR forecasts; the interpretation is analogous for the other forecast pairs. \label{tab:casestudy}}
	\bigskip
	\resizebox{\columnwidth}{!}{%
		\begin{tabular}{lcrrr@{\hskip 0.75cm}rr@{\hskip 0.75cm}rr@{\hskip 0.75cm}rr}
			
			&	&	\multicolumn{3}{c}{Average Brier score} & \multicolumn{2}{c}{HCLR/IDR} & \multicolumn{2}{c}{IDR/HCLR$_{-}$} & \multicolumn{2}{c}{HCLR/HCLR$_{-}$} \\[0.5em]
			
			& Lag & IDR & HCLR & HCLR$_{-}$ & $E$ & $p$ & $E$ & $p$& $E$ & $p$\\[0.25em]
			
			BRU & $1$ & $0.107$ & $0.117$ & $0.118$ & $0$ & $0.9998$ & $> 100$ & $<10^{-4}$ & $> 100$ & $0.0702$\\ 
			& $2$ & $0.119$ & $0.123$ & $0.125$ & $0.01$ & $0.9471$ & $> 100$ & $0.0101$ & $13.602$ & $0.0294$\\ 
			& $3$ & $0.134$ & $0.133$ & $0.136$ & $0.425$ & $0.4405$ & $> 100$ & $0.1916$ & $15.185$ & $0.0019$\\ 
			& $4$ & $0.152$ & $0.145$ & $0.148$ & $4.804$ & $0.0138$ & $1.943$ & $0.9358$ & $5.165$ & $0.0074$\\ 
			& $5$ & $0.171$ & $0.161$ & $0.164$ & $16.969$ & $0.0002$ & $0.415$ & $0.9965$ & $3.436$ & $0.0003$\\[0.25em]
			FRA & $1$ & $0.109$ & $0.111$ & $0.114$ & $0$ & $0.7784$ & $> 100$ & $0.0213$ & $> 100$ & $<10^{-4}$\\ 
			& $2$ & $0.114$ & $0.119$ & $0.122$ & $0.054$ & $0.9643$ & $> 100$ & $0.0002$ & $> 100$ & $0.0004$\\ 
			& $3$ & $0.123$ & $0.127$ & $0.132$ & $0.078$ & $0.9352$ & $> 100$ & $0.0001$ & $26.569$ & $<10^{-4}$\\ 
			& $4$ & $0.147$ & $0.144$ & $0.147$ & $2.291$ & $0.0966$ & $9.618$ & $0.5245$ & $5.54$ & $0.0001$\\ 
			& $5$ & $0.166$ & $0.161$ & $0.163$ & $1.526$ & $0.0305$ & $2.362$ & $0.8871$ & $3.227$ & $0.0051$\\[0.25em]
			LHR & $1$ & $0.135$ & $0.138$ & $0.139$ & $0.029$ & $0.8136$ & $14.979$ & $0.1314$ & $2.845$ & $0.3721$\\ 
			& $2$ & $0.138$ & $0.143$ & $0.143$ & $0.188$ & $0.9189$ & $> 100$ & $0.0509$ & $2.868$ & $0.4369$\\ 
			& $3$ & $0.152$ & $0.154$ & $0.155$ & $0.734$ & $0.7549$ & $40.905$ & $0.1394$ & $2.488$ & $0.3400$\\ 
			& $4$ & $0.169$ & $0.167$ & $0.169$ & $1.429$ & $0.2455$ & $1.7$ & $0.5442$ & $1.744$ & $0.0785$\\ 
			& $5$ & $0.186$ & $0.181$ & $0.182$ & $1.577$ & $0.0753$ & $0.379$ & $0.9288$ & $1.118$ & $0.3216$\\[0.25em]
			ZRH & $1$ & $0.104$ & $0.108$ & $0.110$ & $0.003$ & $0.9306$ & $> 100$ & $0.0055$ & $61.747$ & $0.0003$\\ 
			& $2$ & $0.110$ & $0.112$ & $0.114$ & $0.116$ & $0.7219$ & $36.891$ & $0.0304$ & $10.276$ & $0.0001$\\ 
			& $3$ & $0.121$ & $0.118$ & $0.121$ & $1.516$ & $0.0892$ & $31.924$ & $0.4410$ & $5.098$ & $0.0001$\\ 
			& $4$ & $0.138$ & $0.132$ & $0.134$ & $4.069$ & $0.0027$ & $1.276$ & $0.9588$ & $2.771$ & $0.0015$\\ 
			& $5$ & $0.165$ & $0.156$ & $0.159$ & $15.151$ & $<10^{-4}$ & $0.842$ & $0.9978$ & $2.383$ & $0.0002$\\

		\end{tabular}
	}%
\end{table}

Tables \ref{tab:casestudy} and \ref{tab:casestudy2} show the e-values and one-sided p-values for the hypotheses described above, computed separately for each airport and forecast lag. The e-values are not transformed to p-values here. For interpretation, \citet[Section 3]{Vovk2020a} suggest the discrete scale such that e-values in $(0,1]$, $(1, 3.16]$, $(3.16, 10]$, $(10, 31.6]$, $(31.6, 100]$, and $(100, \infty)$ represent no, poor, substantial, strong, very strong, and decisive evidence against the null hypothesis, respectively. E-values larger than $100$ are not displayed to improve readability, but an untruncated variant of Table \ref{tab:casestudy} is contained in \ref{app:application} so that it is possible to update the e-values with more recent data. For all hypotheses, the p-values and e-values largely lead to the same conclusions. HCLR does not outperform IDR for PoP forecasts at lags 1 to 3, but for the airports Brussels and Zurich there is substantial to strong evidence that it achieves lower Brier scores at the lags 4 and 5. HCLR$_{-}$ is clearly outperformed by the more complex variant with the ensemble-dependent scale parameter at short lags, and also for the longer lead times there is some evidence that including the scale parameter improves the forecasts, except for London airport. As for the difference between IDR and HCLR$_{-}$, both the e-values and the p-values suggest that IDR yields the better forecasts at lags 1 to 3, but at lags 4 and 5, there are no rejections of the null hypothesis. Figure \ref{fig:case} shows how the cumulative products of the e-values for the hypotheses tests at lag 1 evolve over time. If the goal was to accumulate strong evidence against the hypotheses, say exceeding the level $10$, then the hypothesis that IDR outperforms HCLR$_{-}$ could already be rejected with only $9\%$ or $27\%$ of the data, respectively, which is where the corresponding lines first cross the level $10$. For Zurich airport, rejection happens at $85\%$ of the total sample size.

Interestingly, in the comparison of HCLR and HCLR$_{-}$ for Brussels, lag 1, the p-value is non-significant $(0.07$) but the e-value gives decisive evidence ($> 100$). We attribute this to the different null hypotheses of the tests: The mean difference in Brier score is only $0.001$ with an estimated standard deviation of $0.03$, giving only little evidence against the null hypothesis of the Diebold-Mariano test. However, the null hypothesis for the e-value is smaller, requiring that HCLR$_{-}$ outperforms HCLR \emph{at all time points}. Even if the score differences are only small, evidence eventually accumulates over the whole time period; see the rightmost panel of Figure \ref{fig:case}. The fact the e-values in the comparison HCLR/HCLR$_{-}$ decrease with the forecast lag is an effect of the less powerful merging method for e-values with higher lag.

\begin{figure}[t]
	\centering
	\caption{E-values for the hypotheses tests at lag 1 for Brussels (dots), Frankfurt (triangles), London (squares), and Zurich (crosses). The abbreviation of the hypotheses is as in Table \ref{tab:casestudy}. \label{fig:case}}
	\bigskip
	\includegraphics[width=0.9\textwidth]{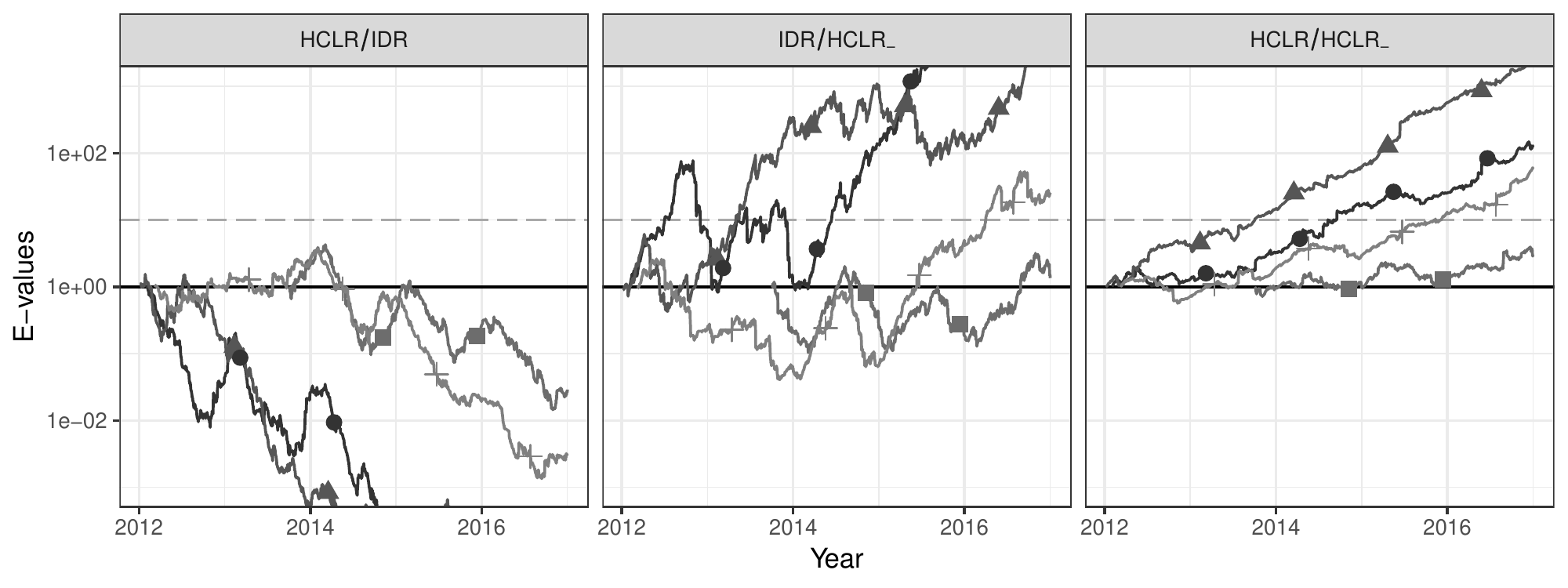}
\end{figure}

In the comparisons of extreme precipitation warnings, the p-value gives some evidence against the null hypothesis for Brussels airport, and the corresponding e-value is decisive, $E = 3703$. For the other lag 1 forecasts, both p-values and e-values do not indicate that including the ensemble standard deviation brings a benefit. As for the higher lags, for London and Zurich airport there is no evidence that HCLR outperforms HCLR$_{-}$, and for Brussels and Frankfurt airport there is only evidence at lags 2 and 3. Overall, the evidence in favour of the HCLR model for issuing extreme precipitation warnings as compared to HCLR$_{-}$ is surprisingly weak.

\begin{table}
	\centering
	\caption{Sample sizes, e-values and p-values for the comparison of tail probability forecasts. The sample size is the number of observations where the condition $\min(p_t, q_t) \geq 0.5$ holds. \label{tab:casestudy2}}
	\bigskip
	\resizebox{\columnwidth}{!}{%
		\begin{tabular}{lllllllll}
			
			& \multicolumn{2}{l}{Brussels} & \multicolumn{2}{l}{Frankfurt} & \multicolumn{2}{l}{London} & \multicolumn{2}{l}{Zurich} \\[0.5em]
			
			Lag & $n$ & $E \ (p)$ & $n$ & $E \ (p)$ & $n$ & $E \ (p)$ & $n$ & $E \ (p)$ \\[0.25em]
			$1$ & $116$ & $> 100 \ (0.050) \quad$ & $79$ & $0.175 \ (0.814) \quad$ & $72$ & $0.45 \ (0.724) \quad$ & $92$ & $0.047 \ (0.892)$\\ 
			$2$ & $88$ & $23.409$ & $87$ & $3.327$ & $69$ & $1.332$ & $99$ & $2.961$\\ 
			$3$ & $68$ & $10.704$ & $62$ & $3.542$ & $60$ & $1.429$ & $75$ & $0.567$\\ 
			$4$ & $49$ & $2.338$ & $53$ & $1.166$ & $39$ & $0.868$ & $52$ & $0.773$\\ 
			$5$ & $28$ & $1.029$ & $26$ & $1.033$ & $30$ & $1.077$ & $36$ & $1.073$\\ 
		\end{tabular}
	}%
\end{table}

\section*{Acknowledgement}
A.~Henzi and J.~F.~Ziegel gratefully acknowledge financial support from the Swiss National Science Foundation. The authors are grateful to Ruodu Wang for introducing them to e-values for hypothesis testing, and to Tobias Fissler and Aaditya Ramdas for helpful comments.

\bibliographystyle{plainnat}
\bibliography{VSIFP_biblio_03}

\begin{thebibliography}{32}
\providecommand{\natexlab}[1]{#1}
\providecommand{\url}[1]{\texttt{#1}}
\expandafter\ifx\csname urlstyle\endcsname\relax
  \providecommand{\doi}[1]{doi: #1}\else
  \providecommand{\doi}{doi: \begingroup \urlstyle{rm}\Url}\fi

\bibitem[Buizza et~al.(2005)Buizza, Houte\-kamer, Pellerin, Toth, Zhu, and
  Wei]{Buizza2005}
R.~Buizza, P.~L. Houte\-kamer, G.~Pellerin, Z.~Toth, Y.~Zhu, and M.~Wei.
\newblock {A comparison of the ECMWF, MSC, and NCEP global ensemble prediction
  systems}.
\newblock \emph{Mon. Weather Rev.}, 133:\penalty0 1076--1097, 2005.

\bibitem[Diebold and Mariano(1995)]{Diebold1995}
F.~X. Diebold and R.~S. Mariano.
\newblock Comparing predictive accuracy.
\newblock \emph{J. Bus. Econ. Stat.}, 13:\penalty0 253--263, 1995.

\bibitem[Ehm and Kr\"{u}ger(2018)]{Ehm2018}
W.~Ehm and F.~Kr\"{u}ger.
\newblock Forecast dominance testing via sign randomization.
\newblock \emph{Electron. J. Stat.}, 12:\penalty0 3758--3793, 2018.

\bibitem[Ehm et~al.(2016)Ehm, Gneiting, Jordan, and Kr\"{u}ger]{Ehm2016}
W.~Ehm, T.~Gneiting, A.~Jordan, and F.~Kr\"{u}ger.
\newblock Of quantiles and expectiles: consistent scoring functions, {C}hoquet
  representations and forecast rankings.
\newblock \emph{J. R. Stat. Soc. Ser. B. Stat. Methodol.}, 78:\penalty0
  505--562, 2016.

\bibitem[Giacomini and White(2006)]{Giacomini2006}
R.~Giacomini and H.~White.
\newblock Tests of conditional predictive ability.
\newblock \emph{Econometrica}, 74:\penalty0 1545--1578, 2006.

\bibitem[Gneiting(2011)]{Gneiting2011}
T.~Gneiting.
\newblock Making and evaluating point forecasts.
\newblock \emph{J. Amer. Statist. Assoc.}, 106:\penalty0 746--762, 2011.

\bibitem[Gneiting and Raftery(2007)]{Gneiting2007a}
T.~Gneiting and A.~E. Raftery.
\newblock Strictly proper scoring rules, prediction, and estimation.
\newblock \emph{J. Amer. Statist. Assoc.}, 102:\penalty0 359--378, 2007.

\bibitem[Gneiting and Ranjan(2013)]{Gneiting2013}
T.~Gneiting and R.~Ranjan.
\newblock Combining predictive distributions.
\newblock \emph{Electron. J. Stat.}, 7:\penalty0 1747--1782, 2013.

\bibitem[Gneiting et~al.(2007)Gneiting, Balabdaoui, and Raftery]{Gneiting2007}
T.~Gneiting, F.~Balabdaoui, and A.~E. Raftery.
\newblock Probabilistic forecasts, calibration and sharpness.
\newblock \emph{J. R. Stat. Soc. Ser. B Stat. Methodol.}, 69:\penalty0
  243--268, 2007.

\bibitem[Gr{\"u}nwald et~al.(2020)Gr{\"u}nwald, de~Heide, and
  Koolen]{Gruenwald2020}
P.~Gr{\"u}nwald, R.~de~Heide, and W.~M. Koolen.
\newblock Safe testing.
\newblock In \emph{2020 Information Theory and Applications Workshop (ITA)},
  pages 1--54. IEEE, 2020.

\bibitem[Henzi et~al.(2021)Henzi, Ziegel, and Gneiting]{Henzi2021}
A.~Henzi, J.~F. Ziegel, and T.~Gneiting.
\newblock Isotonic distributional regression.
\newblock \emph{J. R. Stat. Soc. Ser. B Stat. Methodol., forthcoming}, 2021.

\bibitem[Kelly~Jr(1956)]{Kelly1956}
J.~L. Kelly~Jr.
\newblock A new interpretation of information rate.
\newblock \emph{Bell System Technical Journal}, 35:\penalty0 917--926, 1956.

\bibitem[Lai et~al.(2011)Lai, Gross, and Shen]{Lai2011}
T.~Z. Lai, S.~T. Gross, and D.~B. Shen.
\newblock Evaluating probability forecasts.
\newblock \emph{Ann. Statist.}, 39:\penalty0 2356--2382, 2011.

\bibitem[Lazarus et~al.(2018)Lazarus, Lewis, Stock, and Watson]{Lazarus2018}
E.~Lazarus, D.~J. Lewis, J.~H. Stock, and M.~W. Watson.
\newblock {HAR} inference: Recommendations for practice.
\newblock \emph{J. Bus. Econ. Stat.}, 36:\penalty0 541--559, 2018.

\bibitem[Lerch et~al.(2017)Lerch, Thorarinsdottir, Ravazzolo, and
  Gneiting]{Lerch2017}
S.~Lerch, T.~L. Thorarinsdottir, F.~Ravazzolo, and T.~Gneiting.
\newblock Forecaster’s dilemma: Extreme events and forecast evaluation.
\newblock \emph{Statist. Sci.}, 32:\penalty0 106--127, 02 2017.

\bibitem[Leutbecher and Palmer(2008)]{Leutbecher2008}
M.~Leutbecher and T.~N. Palmer.
\newblock Ensemble forecasting.
\newblock \emph{J. Comput. Phys.}, 227:\penalty0 3515--3539, 2008.

\bibitem[Messner et~al.(2014)Messner, Mayr, Wilks, and Zeileis]{Messner2014}
J.~W. Messner, G.~J. Mayr, D.~S. Wilks, and A.~Zeileis.
\newblock {Extending extended logistic regression: Extended versus separate
  versus ordered versus censored}.
\newblock \emph{Mon. Weather Rev.}, 142:\penalty0 3003--3014, 2014.

\bibitem[Molteni et~al.(1996)Molteni, Buizza, Palmer, and
  Petroliagis]{Molteni1996}
F.~Molteni, R.~Buizza, T.~N. Palmer, and T.~Petroliagis.
\newblock {The ECMWF ensemble prediction system: {M}ethodology and validation}.
\newblock \emph{Q. J. R. Meteorol. Soc.}, 122:\penalty0 73--119, 1996.

\bibitem[Patton(2020)]{Patton2020}
A.~J. Patton.
\newblock Comparing {P}ossibly {M}isspecified {F}orecasts.
\newblock \emph{J. Bus. Econom. Statist.}, 38:\penalty0 796--809, 2020.

\bibitem[Ramdas et~al.(2020)Ramdas, Ruf, Larsson, and Koolen]{Ramdas2020}
A.~Ramdas, J.~Ruf, M.~Larsson, and W.~Koolen.
\newblock Admissible anytime-valid sequential inference must rely on
  nonnegative martingales.
\newblock \emph{arXiv preprint arXiv:2009.03167}, 2020.

\bibitem[Ranjan and Gneiting(2010)]{Ranjan2010}
R.~Ranjan and T.~Gneiting.
\newblock Combining probability forecasts.
\newblock \emph{J. R. Stat. Soc. Ser. B Stat. Methodol.}, 72:\penalty0 71--91,
  2010.

\bibitem[Schervish(1989)]{Schervish1989}
M.~J. Schervish.
\newblock A general method for comparing probability assessors.
\newblock \emph{Ann. Statist.}, 17:\penalty0 1856--1879, 1989.

\bibitem[Seillier-Moiseiwitsch and Dawid(1993)]{Seillier1993}
F.~Seillier-Moiseiwitsch and A.~P. Dawid.
\newblock On testing the validity of sequential probability forecasts.
\newblock \emph{J. Amer. Statist. Assoc.}, 88:\penalty0 355--359, 1993.

\bibitem[Shafer(2021)]{Shafer2021}
G.~Shafer.
\newblock Testing by betting: A strategy for statistical and scientific
  communication.
\newblock \emph{J. R. Stat. Soc. Ser. A Stat. in Society}, 184:\penalty0
  407--431, 2021.

\bibitem[Vannitsem et~al.(2018)Vannitsem, Wilks, and Messner]{Vannitsem2018}
S.~Vannitsem, D.~S. Wilks, and J.~Messner, editors.
\newblock \emph{Statistical Postprocessing of Ensemble Forecasts}.
\newblock Elsevier, 2018.

\bibitem[Vovk and Wang(2020)]{Vovk2020a}
V.~Vovk and R.~Wang.
\newblock True and false discoveries with independent e-values.
\newblock \emph{arXiv preprint arXiv:2003.00593}, 2020.

\bibitem[Vovk and Wang(2021)]{Vovk2021}
V.~Vovk and R.~Wang.
\newblock E-values: Calibration, combination, and applications.
\newblock \emph{Ann. Statist., forthcoming}, 2021.

\bibitem[Wang and Ramdas(2020)]{Wang2020}
R.~Wang and A.~Ramdas.
\newblock False discovery rate control with e-values.
\newblock \emph{arXiv preprint arXiv:2009.02824}, 2020.

\bibitem[Waudby-Smith and Ramdas(2021)]{Waudabysmith2021}
I.~Waudby-Smith and A.~Ramdas.
\newblock Estimating means of bounded random variables by betting.
\newblock \emph{arXiv preprint arXiv:2010.09686}, 2021.

\bibitem[Winkler(1996)]{Winkler1996}
R.~L. Winkler.
\newblock Scoring rules and the evaluation of probabilities.
\newblock \emph{Test}, 5:\penalty0 1--60, 1996.

\bibitem[Yen and Yen(2021)]{Yen2021}
Y.~Yen and T.~Yen.
\newblock Testing forecast accuracy of expectiles and quantiles with the
  extremal consistent loss functions.
\newblock \emph{Int. J. Forecast.}, 37:\penalty0 733--758, 2021.
\newblock ISSN 0169-2070.

\bibitem[Zhu and Timmermann(2020)]{Zhu2020}
Y.~Zhu and A.~Timmermann.
\newblock Can two forecasts have the same conditional expected accuracy?
\newblock \emph{arXiv preprint arXiv:2006.03238}, 2020.

\end{thebibliography}

\appendix

\section{Proofs}
\subsection{Proofs in Section 3}
\begin{proof}[Proof Theorem \ref{thm:uniqueness}]
	If $E(y)$ is of the stated form, then $E(y) \geq E\{\one(p>q)\} = 1-\lambda \geq 0$, and one can easily verify that $E$ has the given null hypothesis. Assume that $p < q$; the case $p > q$ is analogous. Define $d_{p,q}(y) = \myS(p,y) - \myS(q,y)$ and, for $\pi \in [0,1]$,
	\begin{equation} \label{eq:expscore}
		f(\pi) = \mathbb{E}_{\pi}\{d_{p,q}(Y)\} = (1-\pi) d_{p,q}(0) + \pi d_{p,q}(1).
	\end{equation}
	The elementary score representation \eqref{eq:mixture} and $\nu\{[p,q)\} > 0$ imply that $d_{p,q}(0) < 0 < d_{p,q}(1)$, so $f(\pi)$ is strictly increasing in $\pi$ and equal to zero for some $\pi_0 \in (0,1)$. Let $E = E(y)$ be an e-value under $H_{\myS}$ with alternative $H_{\myS}^c$, i.e. $E(y) \geq 0$ and
	\begin{align}
		& \mathbb{E}_{\pi}\{E(Y)\} = (1-\pi) E(0) + \pi E(1) \leq 1 \iff f(\pi) \leq 0. \label{eq:evalueCond}
	\end{align}
	Condition \eqref{eq:evalueCond} implies that $\mathbb{E}_{\pi}\{E(Y)\} = 1$ if and only if $f(\pi) = 0$, which yields
	\begin{equation} \label{eq:ratio}
		\frac{d_{p,q}(0)}{d_{p,q}(1) - d_{p,q}(0)} = \frac{E(0) - 1}{E(1) - E(0)}.
	\end{equation}
	Rearranging this equation gives $E(1) = 1 - \{1 - E(0)\} \cdot d_{p,q}(1) / d_{p,q}(0)$. It follows from \eqref{eq:evalueCond} and \eqref{eq:ratio} that $E(0) \in (0,1)$, so with $\lambda = 1 - E(0)$, we obtain $E(y) = 1 + \lambda d_{p,q}(y) / |d_{p,q}(0)|$. Similar arguments for the case $p > q$ show that in general,
	\[
	E(y) = 1 + \lambda \frac{d_{p,q}(y)}{|d_{p,q}\{\one(p > q)\}|}. 
	\]
\end{proof}

\begin{proof}[Proof Theorem \ref{thm:grow}]
	All e-values for the given null hypothesis are of the form \eqref{eq:representation}. To find the GROW e-value under the alternative that $Y = 1$ with probability $\pi_1$, we have to maximize
	\[
	\mathbb{E}_{\pi_1}[\log\{E_{p,q; \lambda}(Y)\}]
	= (1-\pi_1)\log\left[1 - \lambda \frac{d_{p,q}(0)}{d_{p,q}\{\one(p > q)\}} \right] + \pi_1 \log\left[1 - \lambda \frac{d_{p,q}(1)}{d_{p,q}\{\one(p > q)\}}\right],
	\]
	where again $d_{p,q}(y) = \myS(p,y) - \myS(q,y)$. Let $p < q$; the case $p > q$ is analogous. Under this assumption $d_{p,q}(0) < 0 < d_{p,q}(1)$, and $g(\lambda) = \mathbb{E}_{\pi_1}[\log\{E_{p,q; \lambda}(Y)\}]$ is continuous in $\lambda$ with $g(0) = 0$ and $\lim_{\lambda \rightarrow 1} g(\lambda) = -\infty$, so a maximum is attained at some $\lambda \in [0,1)$. Define $h = d_{p,q}(1)/d_{p,q}(0) < 0$, so that
	\[
	g(\lambda) = (1-\pi_1) \log (1-\lambda) + \pi_1 \log(1-\lambda h), \quad
	g'(\lambda) = -\frac{1-\pi_1}{1-\lambda} - \pi_1 \frac{h}{1-\lambda h},
	\]
	and $g'(\lambda_0) = 0$ is equivalent to $\lambda_0 = \pi_1 + (1-\pi_1)/h$. By definition of $H_{\myS}$, $\pi_1 \not\in H_{\myS}$ holds if and only if $\mathbb{E}_{\pi_1}\{d_{p,q}(Y)\} > 0$, which is equivalent to $\pi_1 + (1-\pi_1)/h > 0$, so indeed $\lambda_0 > 0$ for all $\pi_1 \not\in H_{\myS}$, and
	\begin{align*}
		E_{p,q; \lambda_0}(0) & = 1 - \lambda_0 = (1-\pi_1)\left(1 - \frac{1}{h}\right) 
		= (1-\pi_1)\frac{d_{p,q}(1) - d_{p,q}(0)}{d_{p,q}(1)}, \\
		E_{p,q; \lambda_0}(1) & = 1 - \lambda_0 \frac{d_{p,q}(1)}{d_{p,q}(0)} = \pi_1 \frac{d_{p,q}(0) - d_{p,q}(1)}{d_{p,q}(0)}.
	\end{align*}
	With $d_{p,q}(y) = \int\one\{p \leq \theta < q\}(\theta - y) \, d\nu(\theta)$, it now follows that
	\[
	\frac{d_{p,q}(1) - d_{p,q}(0)}{d_{p,q}(1)} = \frac{-\nu\{[p,q)\}}{-\nu\{[p,q)\} + \int_{[p,q)} \theta \, d\nu(\theta)} = \frac{1}{1 - \kappa_{\nu}\{[p,q)\}\}}
	\]
	and $1 - h =(d_{p,q}(0) - d_{p,q}(1)) / d_{p,q}(0) = \kappa_{\nu}\{[p,q)\}^{-1} > \pi_1^{-1}$, which gives the desired result.
\end{proof}

\begin{proof}[Proof Theorem \ref{thm:allscores}]
	A direct computation shows that $H = [0, p]$ if $p < q$ and $H = [p, 1]$ if $p > q$, and that $\mathbb{E}_{\pi}\{E^{\pi_1*}_{p,q}(Y)\} \leq 1$ for all $\pi \in H$ and $\mathbb{E}_{\pi}\{E^{\pi_1*}_{p,q}(Y)\} > 1$ for $\pi \not\in H$. The result then follows by Theorem 1 of \cite{Gruenwald2020}, with $W_1$ being the Dirac measure of the point $\{\pi_1\}$.
\end{proof}

\begin{proof}[Proof of Proposition \ref{prop:combination}]
Recall that the process $(Y_t, p_t, q_t, \lambda_t)_{t \in \mathbb{N}}$ is adapted to $\mathfrak{F} = (\mathcal{F}_t)_{t \in \mathbb{N}}$. Let $h > 1$. For $k = 1, \dots, h$, define $I_k(t) = \{k + hs\colon s = 0, \dots \lfloor (t-k)/h\rfloor - 1\}$,
\[
M^{[k]}_t = \prod_{l \in I_k(t)} E_{p_{l},q_{l};\lambda_{l}}(Y_{l+h}), \quad
\mathfrak{F}^{[k]} = \left(\mathcal{F}_{\lfloor \frac{t-k}{h}\rfloor h + k}\right)_{t \in \mathbb{N}},
\]
with $\prod_{\emptyset} := 1$ and $\mathcal{F}_j := \{\Omega, \emptyset\}$ for $j \leq 0$. Then $e_t = \sum_{k=1}^h M^{[k]}_t/h$. For $k = 1, \dots, h$, the process $(M^{[k]}_t)_{t \in \mathbb{N}}$ is a nonnegative supermartingale with respect to $\mathfrak{F}^{[k]}$ for any $\mathbb{Q} \in \mathcal{H}_{\myS}$, and therefore satisfies $\mathbb{E}_{\mathbb{Q}}(M^{[k]}_{\tau[k]}) \le 1$ for any $\mathfrak{F}^{[k]}$-stopping time $\tau^{[k]}$. So
\[
\mathbb{E}_\mathbb{Q}\left(\frac{1}{h}\sum_{\ell=1}^h M^{[k]}_{\tau^{[k]}}\right) \le 1.
\]
for $\mathfrak{F}^{[k]}$-stopping times $\tau^{[k]}$, $k = 1,\dots,h$. If $\tau$ is an $\mathfrak{F}$-stopping time, then
\[
\left(\left\lfloor\frac{\tau-k-1}{h}\right\rfloor + 1\right)h + k =: f_k(\tau) \in \{\tau, \dots, \tau + h - 1\}
\]
is an $\mathfrak{F}^{[k]}$-stopping time. To see this, let $t = k + hs + j$ for $s \in \mathbb{N}_0$, $k \in \{1, \dots, k\}$, $j \in \{0, \dots, h-1\}$. Then $\lfloor(t-k)/h\rfloor h+k = k + hs$, and $f_k(\tau) \leq t$ if and only if $\tau \leq k+hs$, so
\[
\{f_{k}(\tau) \leq t\} = \{\tau \leq k + hs\} \in \mathcal{F}_{k+hs} = \mathcal{F}_{\lfloor \frac{t-k}{h}\rfloor h + k}.
\]
This implies that for any $\mathfrak{F}$-stopping time $\tau$, we obtain
\[
\mathbb{E}_{\mathbb{Q}}(M_{\tau + h - 1}) = \mathbb{E}_\mathbb{Q}\left(\frac{1}{h}\sum_{k=1}^h M^{[k]}_{f_k(\tau)}\right) \le 1, \quad \mathbb{Q} \in \mathcal{H}_{\myS},
\]
using the fact that $M_{t + h - 1} = \sum_{k=1}^h M^{[k]}_{f_k(t)}/h$ for $t \in \mathbb{N}$.
\end{proof}


\subsection{Optional stopping for lags $\mathbf{h > 1}$} \label{app:stopping}
In Section \ref{sec:sequential}, the stopping rule
\[
	\tau_{\alpha, h} = \min\Big(T, \, \inf\big[t \geq h + 1: e_{t} \geq \max_{j = t - h + 1, \dots, t - 1} \!\!\!\! E_{p_{j}, q_{j}; \lambda_j}\{\one(p_{j} > q_{j})\}^{-1}  /\alpha\big]\Big),
\]
is defined for e-values of the form
\[
	e_T = \frac{1}{h}\sum_{k = 1}^{h} \prod_{l \in I_k(T)}E_{p_l, q_l; \lambda_l}(Y_{l + h}),
\]
where $I_k(T) = \{k + hs: \, s = 0, \dots, \lfloor (T-k)/h\rfloor - 1\}$. Assume that at time $t$, it is observed that $e_{t} \geq \max_{j = t - h + 1, \dots, t - 1} \! E_{p_{j}, q_{j}; \lambda_j}\{\one(p_{j} > q_{j})\}^{-1}/\alpha$, and that optional stopping is applied, i.e.~$E_{p_s, q_s; \lambda_s}(Y_{t + s}) \equiv 1$ for $s \geq t$. The claim is that then $e_{t + h - 1} \geq 1/\alpha$ no matter what values $Y_{t + 1}, \dots, Y_{t + h - 1}$ take. Because $E_{p_t, q_t; \lambda_t}(Y_{t + h}) \equiv 1$, we have $e_{t + h - 1} = e_{t + h}$. For $k = 1, \dots, h$, let $s_k = k + h\lfloor(t-k)/h\rfloor$, so that $\{s_1, \dots, s_h\} = \{t - h + 1, \dots, t\}$. Then, using that $I_k(t + h) \setminus \{s_k\} = I_k(t)$,
\begin{align*}
	e_{t + h - 1} = e_{t + h} & = \frac{1}{h}\sum_{k = 1}^{h} \left\{E_{p_{s_k}, q_{s_k}; \lambda_{s_k}}(Y_{s_k + h})\prod_{l \in I_k(t + h) \setminus \{s_k\}}E_{p_l, q_l; \lambda_l}(Y_{l + h}) \right\} \\
	& \geq \frac{1}{h}\sum_{k = 1}^{h} \left[E_{p_{s_k}, q_{s_k}; \lambda_{s_k}}\{\one(p_{s_k} > q_{s_k})\}\prod_{l \in I_k(t)}E_{p_l, q_l; \lambda_l}(Y_{l + h}) \right] \\
	& \geq \min_{j = t - h + 1, \dots, t - 1} E_{p_j, q_j; \lambda_j}\{\one(p_j > q_j)\} \cdot \frac{1}{h}\sum_{k = 1}^{h} \prod_{l \in I_k(t)}E_{p_l, q_l; \lambda_l}(Y_{l + h}) \\
	& = \left[\max_{j = t - h + 1, \dots, t - 1} E_{p_j, q_j; \lambda_j}\{\one(p_j > q_j)\}^{-1}\right]^{-1} e_{t} \ \geq \ 1/\alpha.
\end{align*}

\section{Simulation examples: Additional figures} \label{app:simulation}

The simulation example in Section 4.1 in the article has been tested for 
robustness with respect to various parameters:
\begin{itemize}
\item[(i)] Significance levels: $\alpha$: $0.001$, $0.01$, $0.05$
\item[(ii)] Scoring functions: Brier score, spherical score, logarithmic score
\item[(iii)] Sample sizes: $150$, $300$, $600$, $1200$, $2400$
\item[(iv)] Tests for computing p-values: Student's t-test, Wilcoxon's signed rank test
\item[(v)] Alternative hypotheses for e-values: parameter $k$ (as explained in 
  Section 4.1 in the article)
\end{itemize}

For the spherical and the logarithmic score, the probability $\pi_t$ was 
computed in such a way that $\mu = 0.5$ corresponds to a score difference of zero, namely, with
$r_t = \mathbb{E}_{\nu}\big\{\theta \mid \theta \in [\min(p_t,q_t), \max(p_t,q_t))\big\}$,
we set $\pi_t = p_t$ for $\mu = 0$, $\pi_t = r_t$ for $\mu = 0.5$, $\pi_t = q_t$ for $\mu = 1$, and interpolate linearly in between these three points for the other $\mu$.

\medskip\noindent
Figure \ref{fig:sfig1} demonstrates that the rejection rates of the e-values are almost 
the same for all scoring functions.

\medskip\noindent
Figure \ref{fig:sfig2} shows how the rejection rates
vary with the alternative hypothesis for the e-value. In particular, it can be seen that 
the alternative $\pi_t$ is superior and $q_t$ is inferior for all sample sizes and significance levels. As for the alternatives with the parameter $k$, smaller $k$ give
higher rejection rates for small sample sizes and lower rejection rates for larger
samples.

\medskip\noindent
Figure \ref{fig:sfig3} shows that also the rejection rates of Student's t-test are essentially equal for the different scoring functions.

\medskip\noindent
In Figure \ref{fig:sfig4}, it can be seen that the rejection rates of Student's t-test and Wilcoxon's signed rank test for this simulation are almost equal.

\medskip\noindent
Figure \ref{fig:sfig5} shows that close to $\mu = 0.05$, Student's t-test under optional stopping has too high rejection rates independent of the significance level and the sample size.

\bigskip\noindent
The simulation example in Section 4.2 was tested with different significance levels and scoring functions.

\medskip\noindent
Figure \ref{fig:sfig6} shows that the choice of the scoring function has a minor influence
on the rejection rates for the sample sizes $300$ and $600$, and almost no effect for $1200$ and $2400$.

\medskip\noindent
Figure \ref{fig:sfig7} compares the rejection rates of the Diebold-Mariano test and 
the e-values for different significance levels.

\begin{figure}[ht]
\centering
\caption{Rejection rate of stopped e-value (alternative hypothesis with $k = 1$ as explained in the article) for Brier score (dots), spherical score (squares), logarithmic score (triangles), and different significance levels (columns) and sample sizes (rows). \label{fig:sfig1}}
\bigskip
\includegraphics[width=0.9\textwidth]{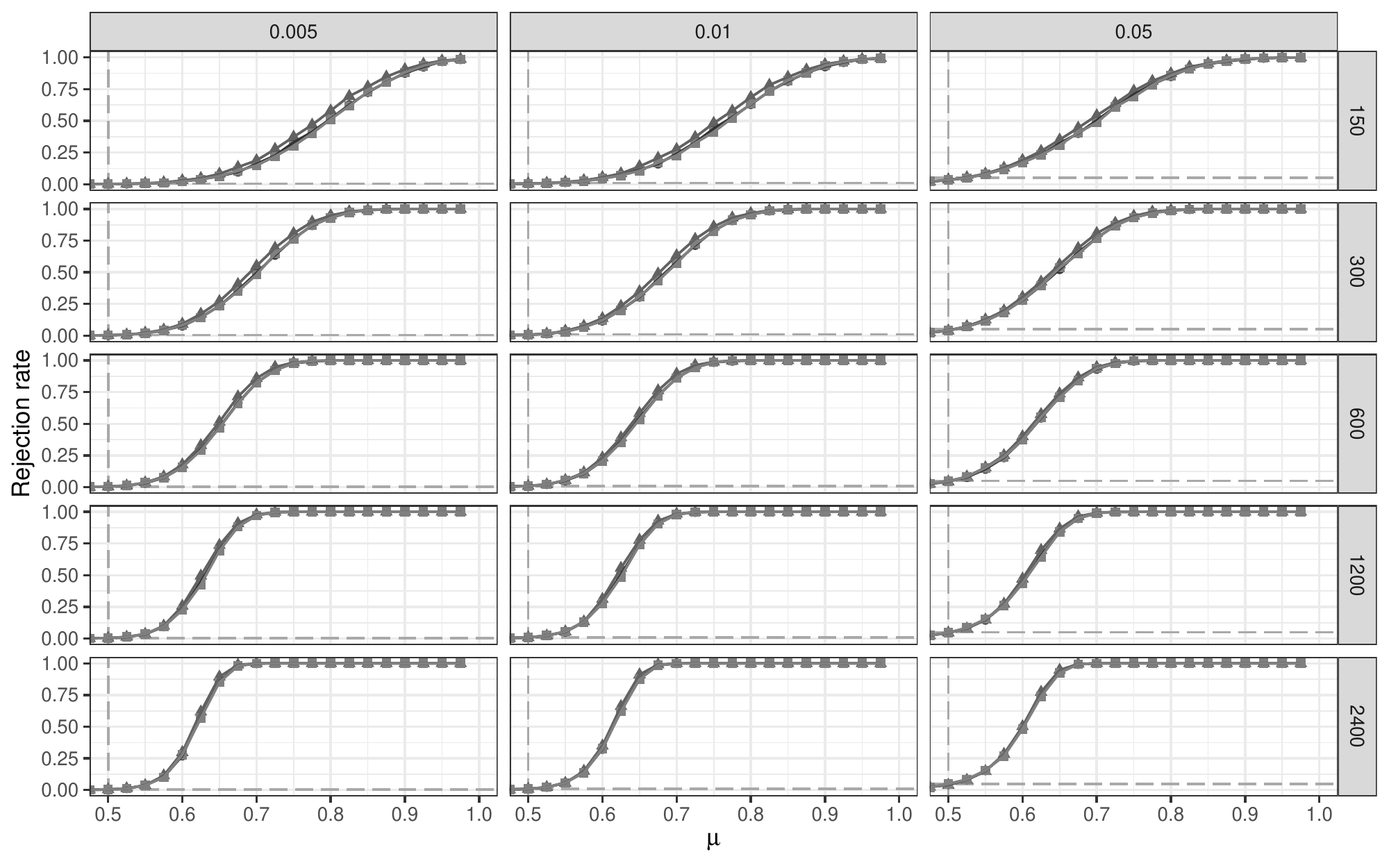}
\end{figure}

\begin{figure}[ht]
\centering
\caption{Rejection of stopped e-values based on Brier score for different alternative hypotheses and different sample sizes and significance levels. The alternatives are $\pi_t$ (dots), $q_t$ (triangles), $k = 1$ (filled squares), $k = 3$ (crosses), $k = 5$ (squares with cross). \label{fig:sfig2}}
\bigskip
\includegraphics[width=0.9\textwidth]{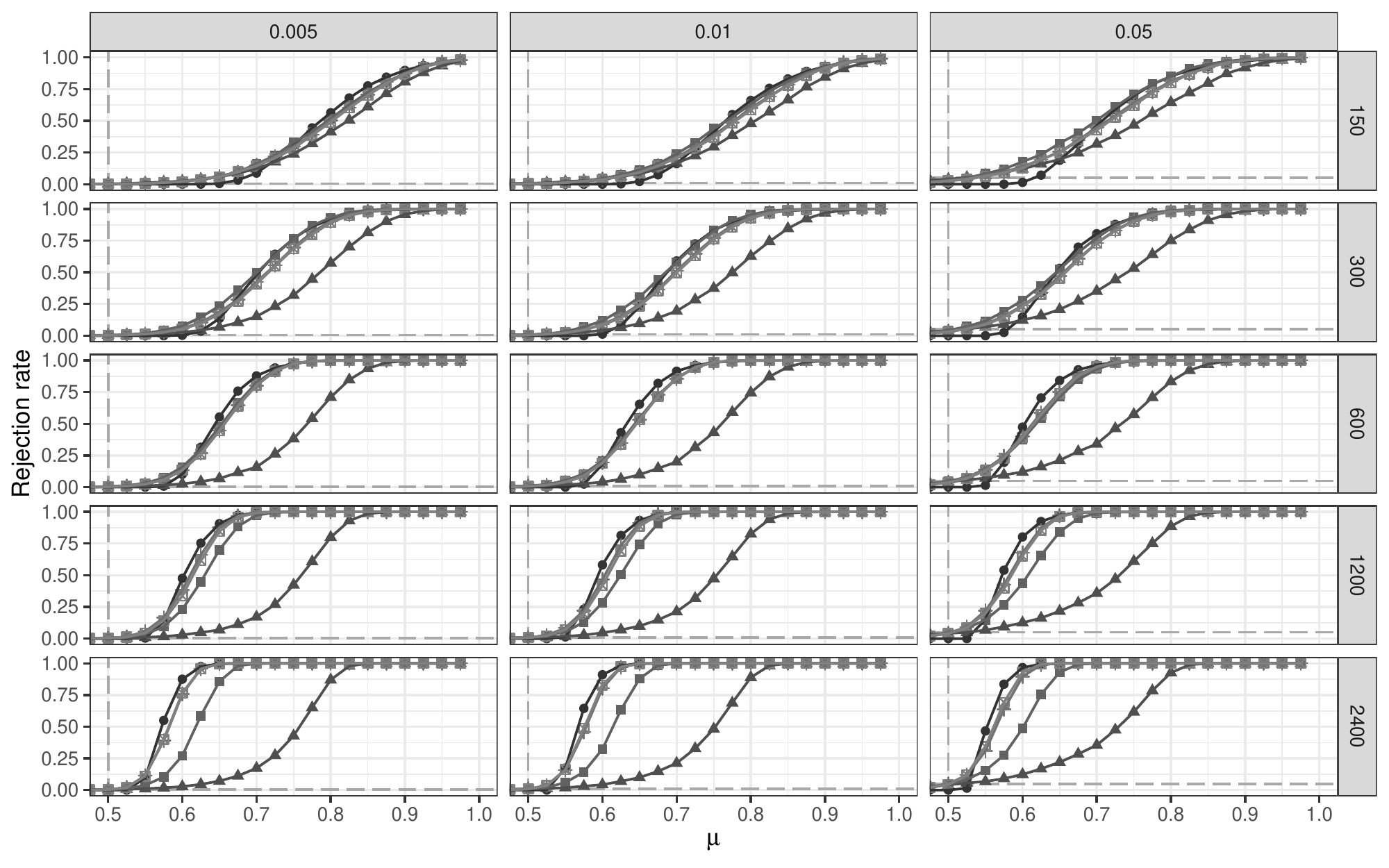}
\end{figure}

\begin{figure}[ht]
\centering
\caption{Rejection rate of Student's t-test for Brier score (dots), spherical score (squares), and logarithmic score (triangles) differences, for different significance levels and sample sizes. \label{fig:sfig3}}
\bigskip
\includegraphics[width=0.9\textwidth]{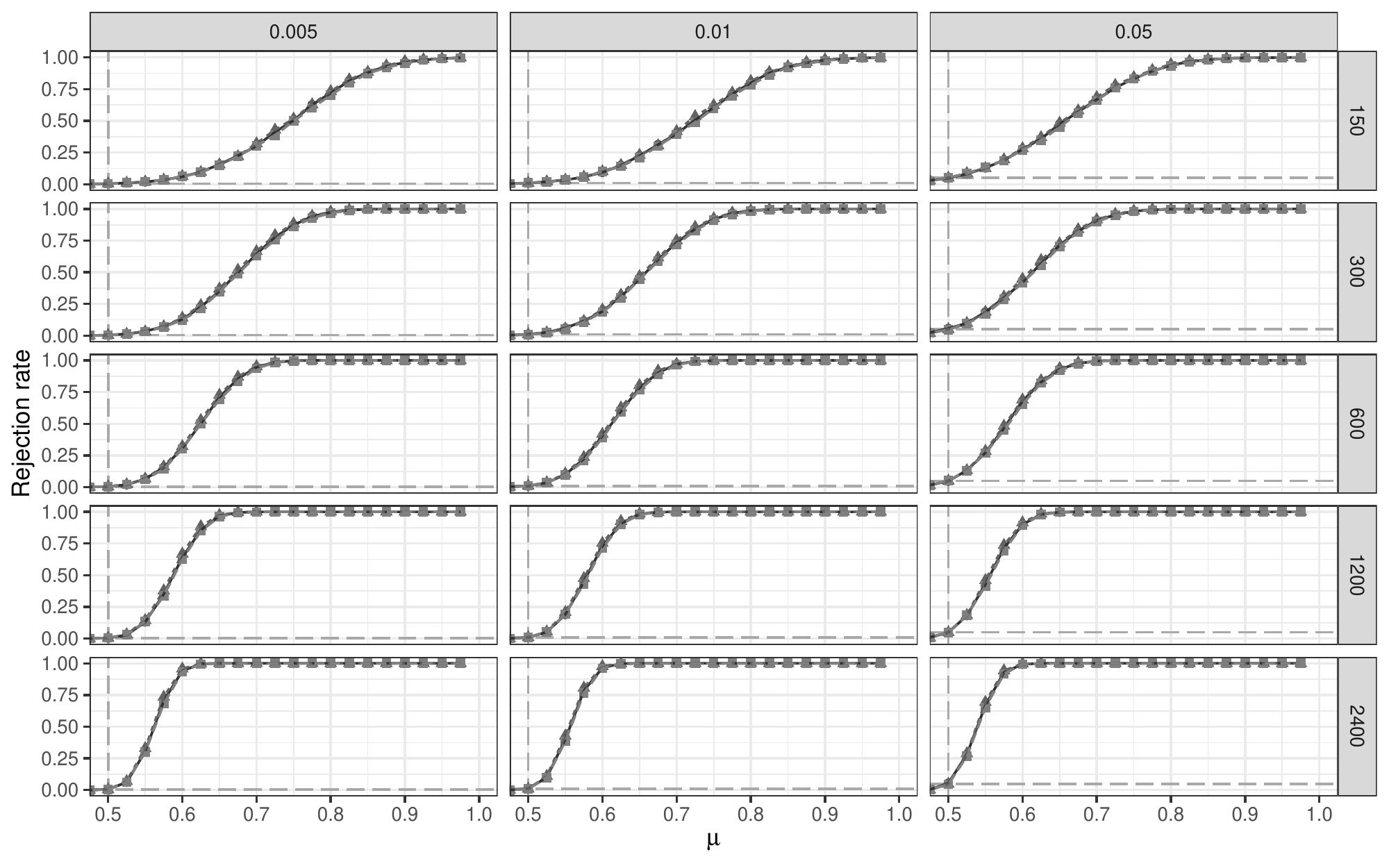}
\end{figure}

\begin{figure}[ht]
\centering
\caption{Rejection rates of Student's t-test (circles) and Wilcoxon's signed rank test (triangles) for Brier score differences, for different significance levels and sample sizes. \label{fig:sfig4}}
\bigskip
\includegraphics[width=0.9\textwidth]{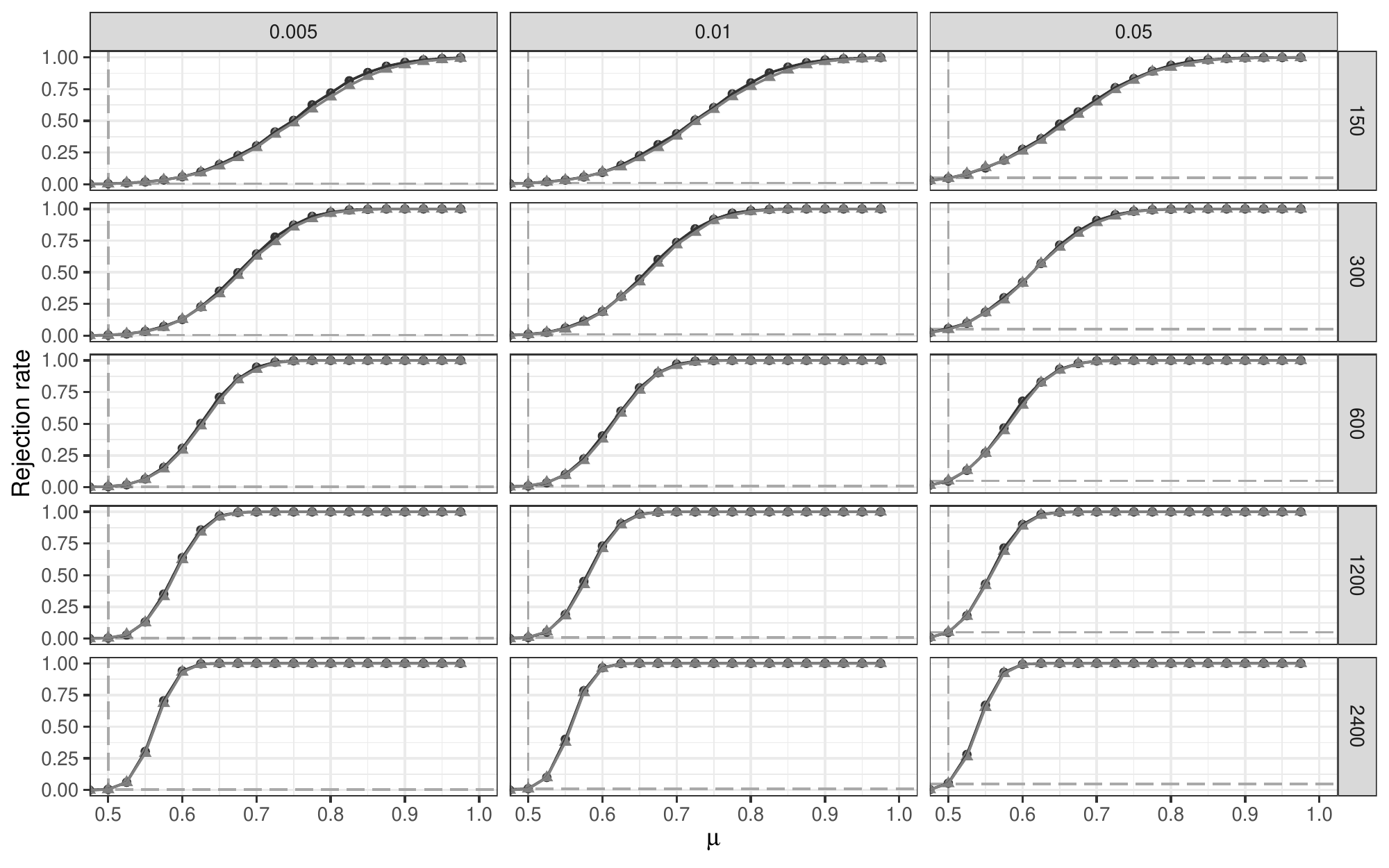}
\end{figure}

\begin{figure}[ht]
\centering
\caption{Rejection rates of Student's t-test under optional stopping, for different significance levels and sample sizes. Optional stops are included at $1$ (triangles), $3$ (squares) and $5$ equispaced time points in between $1$ and the sample size $T$. Dots show the rejection rates without optional stopping. \label{fig:sfig5}}
\bigskip
\includegraphics[width=0.9\textwidth]{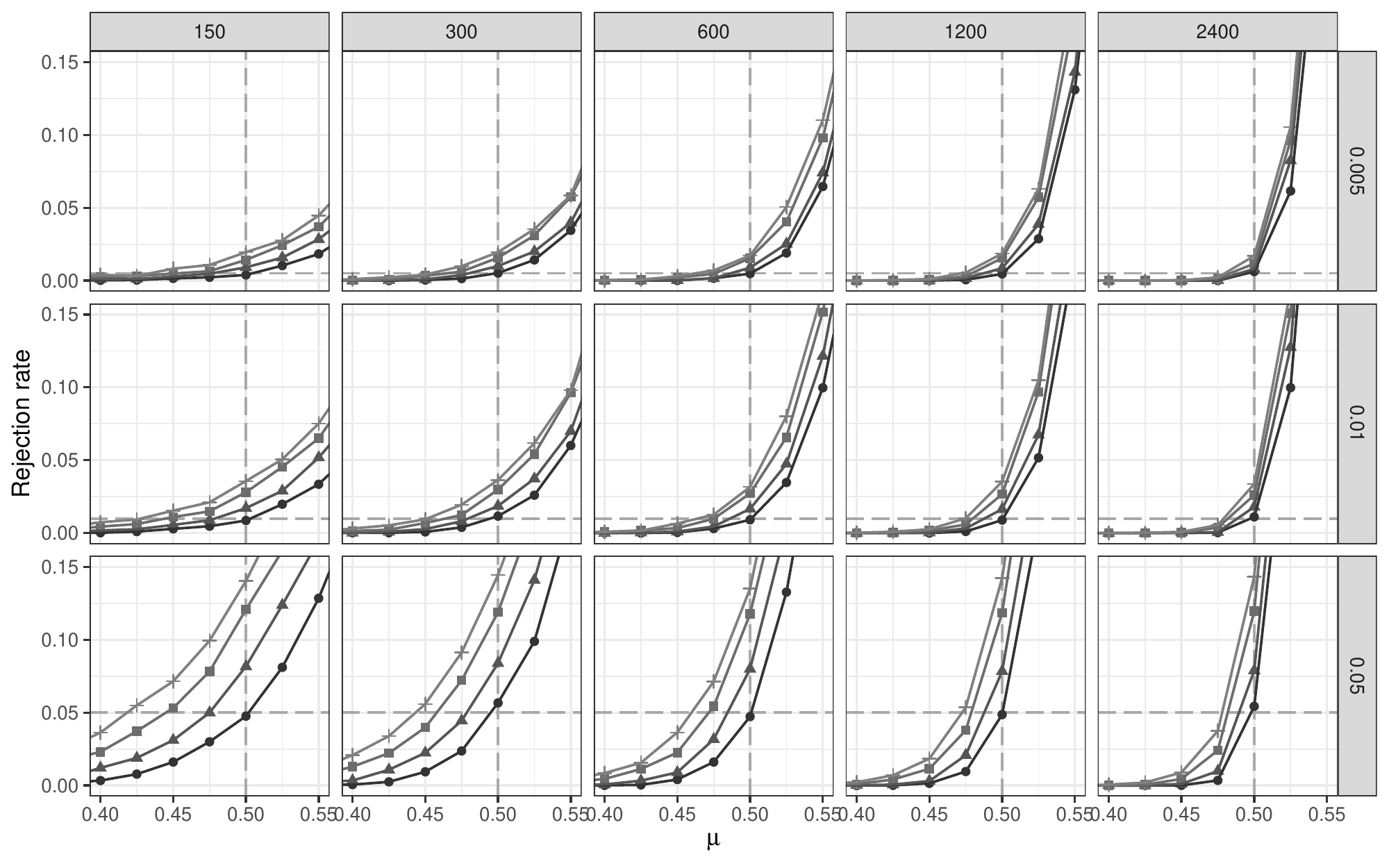}
\end{figure}

\begin{figure}[ht]
\centering
\caption{Rejection rates of Diebold-Mariano test (dashed lines) and e-values (normal lines) for the Brier score (dots), spherical score (squares), and the logarithmic score (triangles), and for different lags (columns) and sample sizes (rows). \label{fig:sfig6}}
\bigskip
\includegraphics[width=0.9\textwidth]{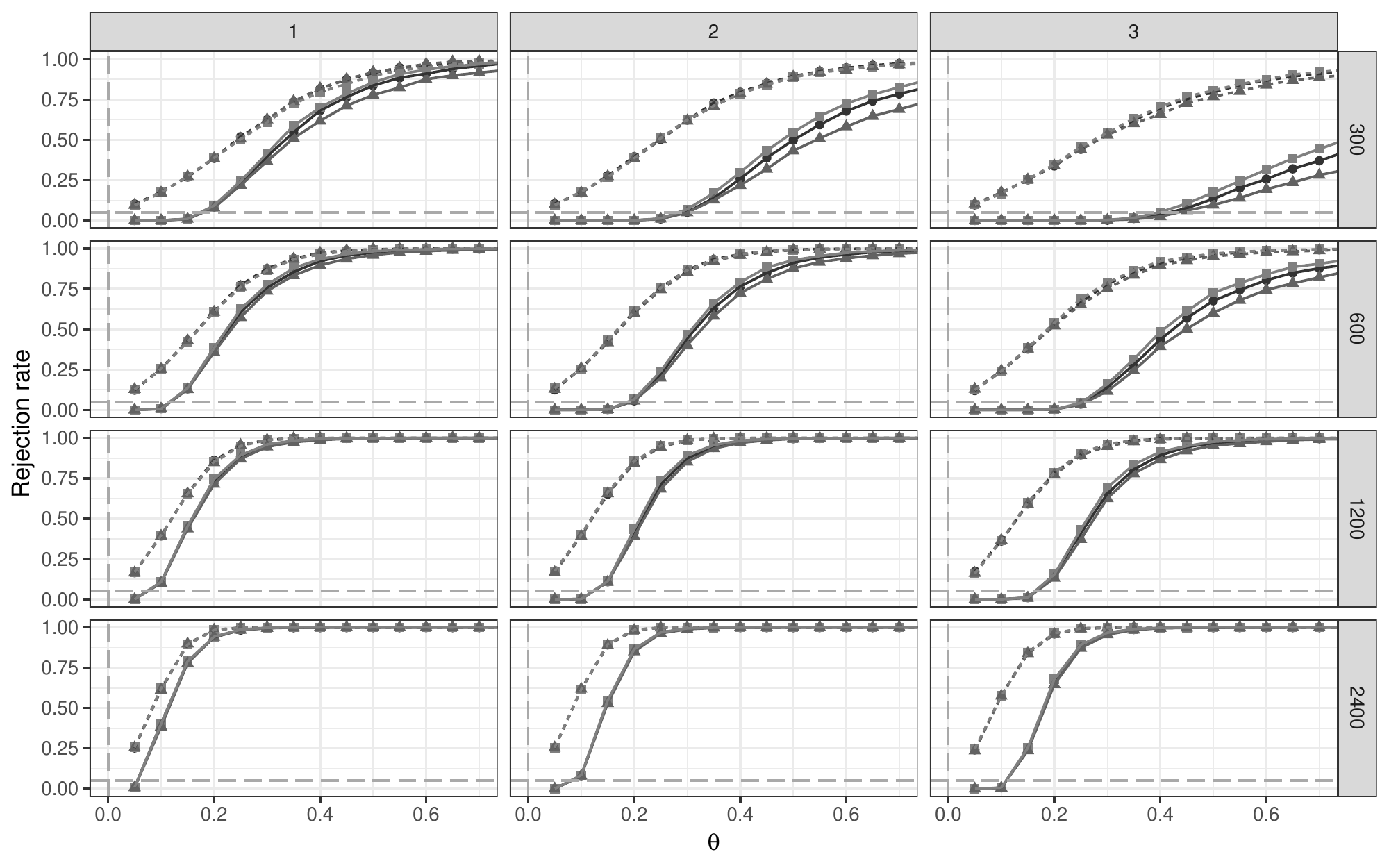}
\end{figure}

\begin{figure}[ht]
\centering
\caption{Rejection rates of the Diebold-Mariano test and E-values for the significance levels $0.005$ (dots), $0.01$ (triangles), and $0.05$ (squares), based on Brier score differences and a sample size of $600$. \label{fig:sfig7}}
\bigskip
\includegraphics[width=0.9\textwidth]{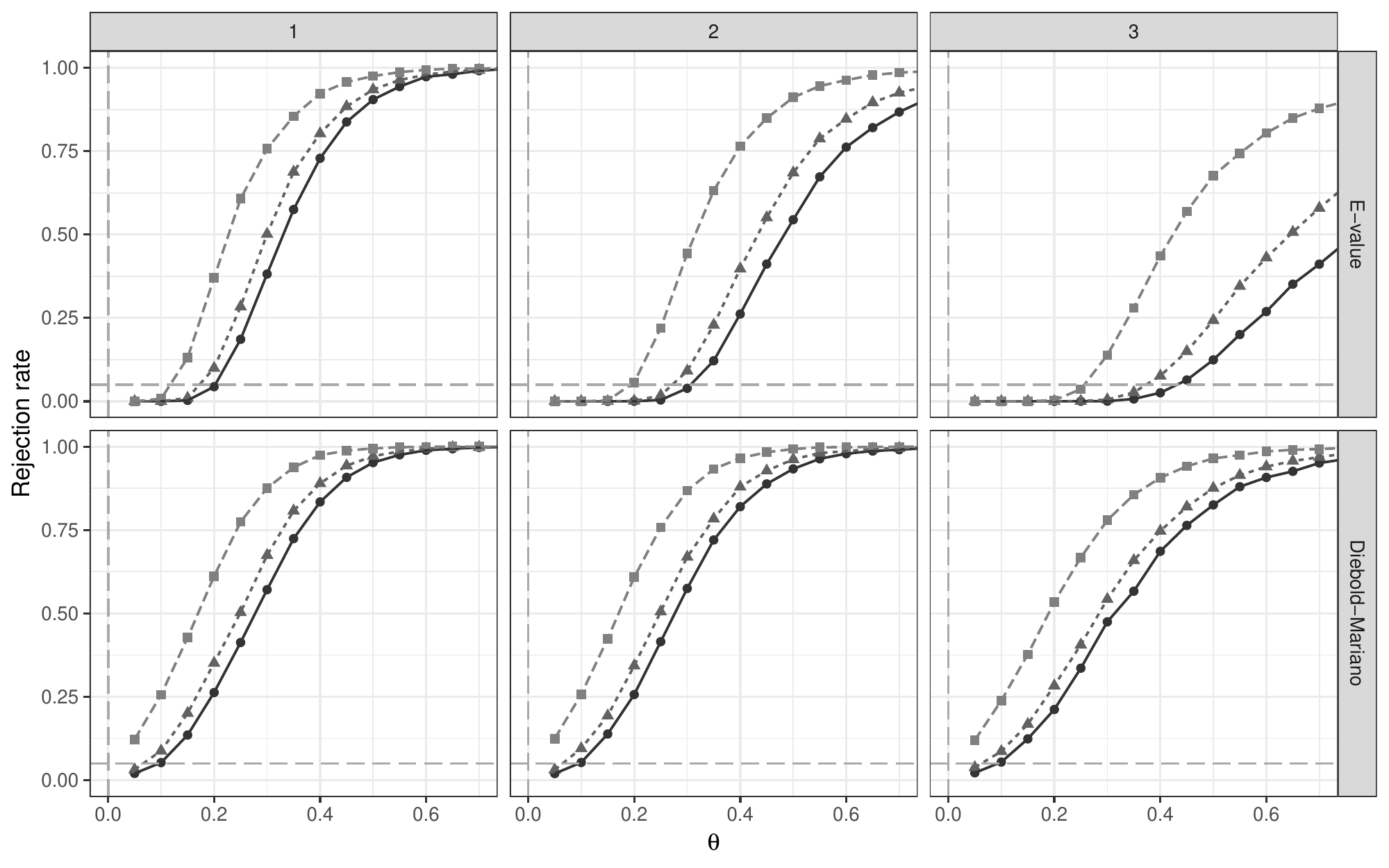}
\end{figure}

\section{Case study: Additional material} \label{app:application}
Table \ref{tab:casestudy_scientific} contains the e-vales and p-values of Table \ref{tab:casestudy} in scientific digit notation.

\begin{table}
\centering
\caption{Brier scores for different probability of precipitation forecasting methods, and e-values ($E$) and p-values ($p$) for testing significance of score differences. The columns HCLR/IDR show e-values and p-values for tests tests of the null hypothesis that IDR probability of precipitation forecasts achieve a lower Brier score the HCLR forecasts; the interpretation is analogous for the other forecast pairs. \label{tab:casestudy_scientific}}
\bigskip
\resizebox{\columnwidth}{!}{%
\begin{tabular}{lcrrrrrrrrr}

&	&	\multicolumn{3}{c}{Brier score} & \multicolumn{2}{c}{HCLR/IDR} & \multicolumn{2}{c}{IDR/HCLR$_{-}$} & \multicolumn{2}{c}{HCLR/HCLR$_{-}$} \\[0.5em]

& Lag & IDR & HCLR & HCLR$_{-}$ & $E$ & $p$ & $E$ & $p$& $E$ & $p$\\[0.25em]

BRU & $1$ & $0.107$ & $0.117$ & $0.118$ & $5.6\text{e}\!-\!08$ & $1.0\text{e+}00$ & $5.0\text{e+}09$ & $1.3\text{e}\!-\!05$ & $1.3\text{e+}02$ & $7.0\text{e}\!-\!02$ \\

& $2$ & $0.119$ & $0.123$ & $0.125$ & $9.5\text{e}\!-\!03$ & $9.5\text{e}\!-\!01$ & $2.2\text{e+}02$ & $1.0\text{e}\!-\!02$ & $1.4\text{e+}01$ & $2.9\text{e}\!-\!02$ \\

& $3$ & $0.134$ & $0.133$ & $0.136$ & $4.3\text{e}\!-\!01$ & $4.4\text{e}\!-\!01$ & $5.4\text{e+}02$ & $1.9\text{e}\!-\!01$ & $1.5\text{e+}01$ & $1.9\text{e}\!-\!03$ \\

& $4$ & $0.152$ & $0.145$ & $0.148$ & $4.8\text{e+}00$ & $1.4\text{e}\!-\!02$ & $1.9\text{e+}00$ & $9.4\text{e}\!-\!01$ & $5.2\text{e+}00$ & $7.4\text{e}\!-\!03$ \\

& $5$ & $0.171$ & $0.161$ & $0.164$ & $1.7\text{e+}01$ & $2.3\text{e}\!-\!04$ & $4.1\text{e}\!-\!01$ & $1.0\text{e+}00$ & $3.4\text{e+}00$ & $3.3\text{e}\!-\!04$ \\[0.25em]

FRA & $1$ & $0.109$ & $0.111$ & $0.114$ & $1.4\text{e}\!-\!06$ & $7.8\text{e}\!-\!01$ & $1.6\text{e+}11$ & $2.1\text{e}\!-\!02$ & $2.4\text{e+}03$ & $2.8\text{e}\!-\!06$ \\

& $2$ & $0.114$ & $0.119$ & $0.122$ & $5.4\text{e}\!-\!02$ & $9.6\text{e}\!-\!01$ & $1.3\text{e+}06$ & $2.3\text{e}\!-\!04$ & $2.5\text{e+}02$ & $4.2\text{e}\!-\!04$ \\

& $3$ & $0.123$ & $0.127$ & $0.132$ & $7.8\text{e}\!-\!02$ & $9.4\text{e}\!-\!01$ & $3.8\text{e+}04$ & $1.3\text{e}\!-\!04$ & $2.7\text{e+}01$ & $5.4\text{e}\!-\!06$ \\

& $4$ & $0.147$ & $0.144$ & $0.147$ & $2.3\text{e+}00$ & $9.7\text{e}\!-\!02$ & $9.6\text{e+}00$ & $5.2\text{e}\!-\!01$ & $5.5\text{e+}00$ & $5.9\text{e}\!-\!05$ \\

& $5$ & $0.166$ & $0.161$ & $0.163$ & $1.5\text{e+}00$ & $3.0\text{e}\!-\!02$ & $2.4\text{e+}00$ & $8.9\text{e}\!-\!01$ & $3.2\text{e+}00$ & $5.1\text{e}\!-\!03$ \\[0.25em]

LHR & $1$ & $0.135$ & $0.138$ & $0.139$ & $2.9\text{e}\!-\!02$ & $8.1\text{e}\!-\!01$ & $1.5\text{e+}01$ & $1.3\text{e}\!-\!01$ & $2.8\text{e+}00$ & $3.7\text{e}\!-\!01$ \\

& $2$ & $0.138$ & $0.143$ & $0.143$ & $1.9\text{e}\!-\!01$ & $9.2\text{e}\!-\!01$ & $1.2\text{e+}02$ & $5.1\text{e}\!-\!02$ & $2.9\text{e+}00$ & $4.4\text{e}\!-\!01$ \\

& $3$ & $0.152$ & $0.154$ & $0.155$ & $7.3\text{e}\!-\!01$ & $7.5\text{e}\!-\!01$ & $4.1\text{e+}01$ & $1.4\text{e}\!-\!01$ & $2.5\text{e+}00$ & $3.4\text{e}\!-\!01$ \\

& $4$ & $0.169$ & $0.167$ & $0.169$ & $1.4\text{e+}00$ & $2.5\text{e}\!-\!01$ & $1.7\text{e+}00$ & $5.4\text{e}\!-\!01$ & $1.7\text{e+}00$ & $7.8\text{e}\!-\!02$ \\

& $5$ & $0.186$ & $0.181$ & $0.182$ & $1.6\text{e+}00$ & $7.5\text{e}\!-\!02$ & $3.8\text{e}\!-\!01$ & $9.3\text{e}\!-\!01$ & $1.1\text{e+}00$ & $3.2\text{e}\!-\!01$ \\[0.25em]

ZRH & $1$ & $0.104$ & $0.108$ & $0.110$ & $3.0\text{e}\!-\!03$ & $9.3\text{e}\!-\!01$ & $3.0\text{e+}04$ & $5.5\text{e}\!-\!03$ & $6.2\text{e+}01$ & $3.2\text{e}\!-\!04$ \\

& $2$ & $0.110$ & $0.112$ & $0.114$ & $1.2\text{e}\!-\!01$ & $7.2\text{e}\!-\!01$ & $3.7\text{e+}01$ & $3.0\text{e}\!-\!02$ & $1.0\text{e+}01$ & $5.0\text{e}\!-\!05$ \\

& $3$ & $0.121$ & $0.118$ & $0.121$ & $1.5\text{e+}00$ & $8.9\text{e}\!-\!02$ & $3.2\text{e+}01$ & $4.4\text{e}\!-\!01$ & $5.1\text{e+}00$ & $1.0\text{e}\!-\!04$ \\

& $4$ & $0.138$ & $0.132$ & $0.134$ & $4.1\text{e+}00$ & $2.7\text{e}\!-\!03$ & $1.3\text{e+}00$ & $9.6\text{e}\!-\!01$ & $2.8\text{e+}00$ & $1.5\text{e}\!-\!03$ \\

& $5$ & $0.165$ & $0.156$ & $0.159$ & $1.5\text{e+}01$ & $2.3\text{e}\!-\!05$ & $8.4\text{e}\!-\!01$ & $1.0\text{e+}00$ & $2.4\text{e+}00$ & $1.7\text{e}\!-\!04$ \\
\end{tabular}
}%
\end{table}

\end{document}